%% file: main-usenix.tex
\documentclass[letterpaper,twocolumn,10pt]{article}
\usepackage{usenix-2020-09}

\usepackage{cite}
\usepackage{graphicx}
\usepackage{textcomp}

\usepackage{xspace}
\newcommand{\eg}{\emph{e.g.,}\xspace}
\newcommand{\ie}{\emph{i.e.,}\xspace}
\newcommand{\first}{(i)\xspace}
\newcommand{\second}{(ii)\xspace}
\newcommand{\third}{(iii)\xspace}

\newcommand{\etal}{\emph{et al.}\xspace}
\newcommand\update[1]{{#1}}

\usepackage{amsmath,amsfonts,amssymb}
\usepackage{mathtools}
\usepackage{bm}
\usepackage{amsthm}
\theoremstyle{plain}
\newtheorem{theorem}{Theorem}
\newtheoremstyle{reducespacing}
  {3pt}
  {3pt}
  {}
  {0pt}
  {\bfseries}
  {.}
  { }
  {\thmname{#1}\thmnumber{ #2}\textnormal{\thmnote{ (#3)}}}
\theoremstyle{reducespacing}
\newtheorem{definition}{Definition}
\usepackage{caption}
\usepackage{tabu}
\usepackage{multirow}
\usepackage{multicol}
\usepackage{ctable}
\usepackage{algorithm}
\usepackage[noend]{algpseudocode}
\usepackage{balance}
\usepackage{enumitem}
\usepackage[T1]{fontenc}
\usepackage[utf8]{inputenc}
\usepackage{mathptmx}
\usepackage{pifont}
\newcommand{\cmark}{\ding{51}}
\newcommand{\xmark}{\ding{55}}

\newcommand{\model}{BEAM\xspace}
\newcommand{\modelfull}{BGP Semantics Aware Network Embedding\xspace}
\newcommand{\modelfullcap}{\textbf{B}GP s\textbf{E}m\textbf{A}ntics aware network e\textbf{M}bedding\xspace}

\begin{document}

\date{}

\title{\Large \bf Learning with Semantics: Towards a Semantics-Aware Routing Anomaly Detection System}


\author{
{\rm Yihao Chen}\\
Tsinghua University
\and
{\rm Qilei Yin}\\
Zhongguancun Lab
\and
{\rm Qi Li}\\
Tsinghua University
\and
{\rm Zhuotao Liu}\\
Tsinghua University
\and
{\rm Ke Xu}\\
Tsinghua University
\and
{\rm Yi Xu}\\
Tsinghua University
\and
{\rm Mingwei Xu}\\
Tsinghua University
\and
{\rm Ziqian Liu}\\
China Telecom
\and
{\rm Jianping Wu}\\
Tsinghua University
} 

\maketitle

\begin{abstract}
BGP is the de facto inter-domain routing protocol to ensure global connectivity of the Internet. However, various reasons, such as deliberate attacks or misconfigurations, could cause BGP routing anomalies. Traditional methods for BGP routing anomaly detection require significant manual investigation of routes by network operators. Although machine learning has been applied to automate the process, prior arts typically impose significant training overhead (such as large-scale data labeling and feature crafting), and only produce uninterpretable results. To address these limitations, this paper presents a routing anomaly detection system centering around a novel network representation learning model named \model. The core design of \model is to accurately learn the unique properties (defined as \emph{routing role}) of each Autonomous System (AS) in the Internet by incorporating BGP semantics. As a result, routing anomaly detection, given \model, is reduced to a matter of discovering unexpected routing role churns upon observing new route announcements. We implement a prototype of our routing anomaly detection system and extensively evaluate its performance. The experimental results, based on 18 real-world RouteViews datasets containing over 11 billion route announcement records, demonstrate that our system can detect all previously-confirmed routing anomalies, while only introducing at most five false alarms every 180 million route announcements.
We also deploy our system at a large ISP to perform real-world detection for one month.  During the course of deployment, our system 
detects 497 true anomalies in the wild with an average of only 1.65 false alarms per day. 
\end{abstract}

\input{content/sec1-introduction}
\input{content/sec2-background}
\input{content/sec3-model-method}

\input{content/sec3-model-result}
\input{content/sec4-detection}

\input{content/sec5-detection-result}
\input{content/sec6-case-study}
\input{content/sec7-discussion}
\input{content/sec8-related-work}

\input{content/sec9-conclusion}
\bibliographystyle{IEEEtran}
\bibliography{ref}

\appendix
\input{content/supplementary}

\end{document}

%% file: content/sec1-introduction.tex
\section{Introduction}

The Border Gateway Protocol (BGP) is the de facto inter-domain routing protocol to achieve global connectivity in the Internet. BGP establishes Internet-wide routing paths by exchanging route announcements among the networks operated by different organizations, referred to as Autonomous Systems (ASes). Each route announcement carries AS-level path information for reaching certain prefixes (\ie a block of IP addresses). At its steady state, every AS learns an AS-path to reach every globally routable Internet prefix.
Despite its global adoption, BGP itself has no built-in authentication mechanism.
As a result, a misbehaved AS can announce arbitrary routes in the Internet, either due to deliberate attacks or misconfigurations. These bogus routes form serious threats to routing security, namely BGP hijacking (\ie forcing certain traffic to go through a malicious AS) and BGP route leaks (\ie redirecting traffic over unintended links). Over the past decade, the Internet has witnessed several severe BGP incidents. 
For example, a Swiss company leaked over 70,000 routes in 2019~\cite{european2019} and a British company hijacked more than 31,000 prefixes in 2021~\cite{major2021}. A cryptocurrency platform recently confirmed the loss of \$1.9 million after
a BGP hijacking attack~\cite{klayswap2022}. \update{Although several security extensions have been proposed to counter these threats, \eg BGPsec~\cite{lepinski2017bgpsec}, psBGP~\cite{oorschot2007interdomain} and S-BGP~\cite{kent2000secure}, they are not widely deployed, possibly due to incompatibility with the current Internet architecture. Besides, while RPKI~\cite{mohapatra2013bgp} has gained traction in providing authoritative information about IP prefix ownership, its effectiveness is largely limited by the incomplete deployment of ROV~\cite{chen2022rov}. More importantly, RPKI is not designed to mitigate route manipulation attacks or route leaks. } 

\update{
Detecting routing anomalies in the global Internet is the first step towards secure Internet routing. The community has proposed significant research in this regard}
\cite{zheng2007light,vervier2015mind,sermpezis2018artemis,schlamp2016heap,hu2007accurate,li2012buddyguard,shi2012detecting,zhang2008ispy,li2005internet,yan2009bgpmon}. However, they typically rely on extensive analysis of routing data from multiple sources. More crucially, these methods require non-trivial human supervision to produce reasonable results.
The advance in Machine Learning (ML) motivates the community to apply ML techniques to automate and simplify anomaly detection by recognizing different patterns of route announcements~\cite{cheng2016ms,testart2019profiling,cheng2018multi,dong2021isp,al2015detecting,al2012machine,lutu2014separating,deshpande2009online,theodoridis2013novel,shapira2020deep,hoarau2021suitability,shapira2022ap2vec,sanchez2019comparing}.
However, existing methods require large datasets with manual labels and/or handcrafted features, imposing significant overheads on data collection and model update. Moreover, many of these methods learn deep latent features for classification, producing largely uninterpretable results.
As a consequence, these methods 
provide limited practical guidance for network operators to fix routing anomalies. 

To address these challenges, we present a routing anomaly detection system centering around a novel network representation learning model, \model (\modelfullcap). Instead of 
learning any latent or opaque features, \model enables interpretable and accurate routing anomaly detection based on the intrinsic routing characteristics of ASes that are derived from the \emph{domain specific knowledge of BGP semantics}. Specifically, we propose the concept of \emph{AS routing role} to meaningfully characterize ASes in BGP route announcements. The design of routing role is derived from the AS business relationship
graph (rather than any handcrafted features), because an AS's business relationship with its neighboring ASes determines how the AS chooses to update the route announcements received from neighbors, and how the newly generated route announcements are further propagated~\cite{gao2001inferring}. Given accurate modeling of ASes' routing roles, anomaly detection is reduced to a matter of detecting unexpected AS routing role churns from the original route to the new route announcement.

The key design challenges in obtaining routing roles are two-fold.
First, the available dataset of route announcements could contain non-trivial noises due to the unrevealed routing anomalies in the Internet~\cite{ballani2007study}.
For instance, AS 4134 leaked over 70,000 routes in 2019~\cite{european2019} and AS 55410 hijacked more than 31,000 prefixes in 2021~\cite{major2021}. 
As a result, computing routing roles directly from the noisy route announcement dataset (using either raw announcements or statistical features) could lead to high false positive rates~\cite{testart2019profiling,al2012machine,lutu2014separating,deshpande2009online}. Second, due to the dynamism and scale of Internet routing, AS routing roles are evolving over time. 

To address the above challenges, \model employs a novel embedding mechanism to learn an embedding vector for each AS based on the AS graph constructed from AS relationships. The key of \model's embedding is to preserve an AS's proximity and hierarchy properties that are essential to its routing role. The exact definitions of proximity and hierarchy are given in \S\ref{sec:method:problem-definition}. The embedding vectors are further employed to uniquely represent and interpret the routing roles of ASes, based on which our routing anomaly detection system reports routing anomalies
upon observing abnormal routing role churns. Further, we design our learning mechanism to ensure that the embedding vectors can properly capture routing roles despite the ever-changing Internet routing and topology.

We validate our system on 18 real-world route announcement datasets collected from global vantage points\footnote{A vantage point is a BGP participant (\eg a router) that provides public access to its routing table and/or its received route announcements.}. The entire datasets include over 11 billion route announcement records spanning from 2008 to 2021. The experimental results show that
\model
produces interpretable results regarding AS routing role changes, based on which our system correctly identifies all previously-confirmed routing anomalies, while only incurring at most five false alarms every 180 million route announcements 
\update{(about 1.61 false alarms per day)}. 
We also deploy our system at a large ISP to detect routing anomalies over a one-month period.
We further visualize AS routing roles to achieve interpretable routing behavior analysis.
Our work can serve as a complement to existing BGP security extensions, such as RPKI, to protect against both BGP hijacking and BGP route leaks.

To summarize, our contributions are four-fold: 
\begin{itemize}[itemsep=0em,align=parleft,left=0pt..1em]
    \item We design the first BGP semantics aware network representation learning model, \model, that accurately captures the routing roles of ASes.
    \item We develop an unsupervised routing anomaly detection system based on \model, which achieves real-time detection in an interpretable manner without requiring labeled routing data or feature engineering.
    \item We validate our system by conducting experiments on 18 RouteViews datasets with more than 11 billion route announcement records. Our system can detect all confirmed anomalies with minor false alarms.
    \item We deploy our system at a large ISP to collect real-world detection results for a month. The system detects 497 true anomalies in the wild from over 150 million 
    live route announcements, with only 1.65 daily false alarms on average. 
\end{itemize}


%% file: content/sec2-background.tex
\section{Background}\label{sec:background}

\begin{figure}[t]
    \centering
    \includegraphics[width=\linewidth]{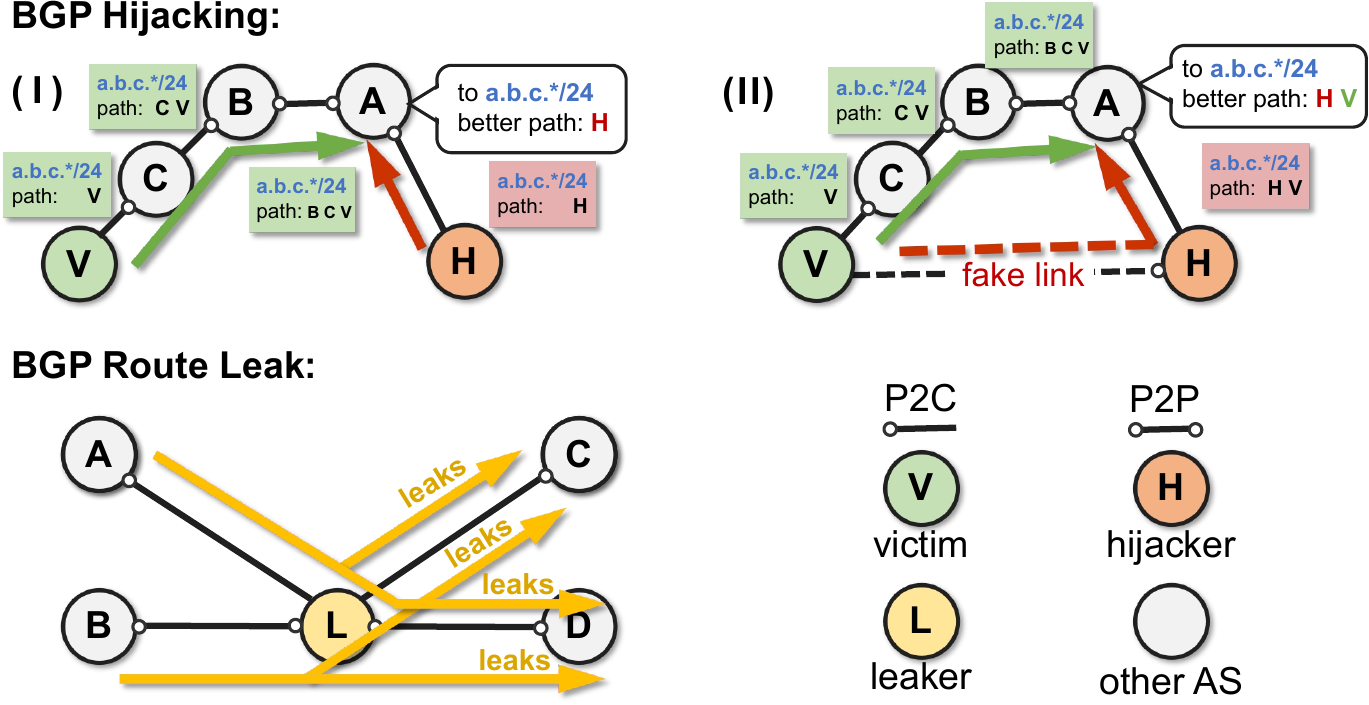}
    \caption{\textbf{Illustrations of BGP anomalies.} \textnormal{In BGP hijacking, the adversary can either (I) falsely claim the ownership of a prefix, or (II) announce a fake yet more preferable route. 
    In BGP route leak, routes are propagated to unintended ASes.
    }}
    \label{fig:bgp-anomaly}
\end{figure}

\noindent\textbf{Inter-domain Routing Protocol.}
The Internet contains over 73,000 advertised Autonomous Systems (ASes) as of Jun 2023.
Each AS consists of several networks under the management of the same organization and is identified by a uniquely allocated non-negative integer called AS Number (ASN). BGP is the de facto inter-AS routing protocol to achieve global connectivity of the Internet. BGP is a path-vector routing protocol that maintains AS-level path information, which gets updated as BGP announcements propagate in the network. Upon receiving a BGP announcement, an AS, following its \emph{routing policy}, may stop further propagating the announcement, or append its ASN to the AS-path and send the updated announcement to a \emph{selective} set of neighbors.

\update{Business relationship
largely determines
one AS's routing policy~\cite{gao2001inferring,prehn2021biased}}. Two neighboring ASes typically have three types of business relationships\footnote{We ignore complex AS relationships~\cite{giotsas2014inferring}. They are much more unusual and play a minor role in defining the overall routing behavior of an AS.}: provider-to-customer (P2C), peer-to-peer (P2P) and customer-to-provider (C2P), where a customer AS pays its provider for connectivity while two peering ASes forward traffic to each other free of charge. Thus, the inter-domain routing system of the Internet can be reconstructed as an AS graph based on AS relationships. This AS-level topology exhibits hierarchy~\cite{gao2001inferring}, with several well recognized Tier-1 (large-scale) ASes. However, the topology is not strictly hierarchical and is flattening over time~\cite{luckie2013relationships}.

\noindent\textbf{BGP Anomalies.}
Although widely deployed, BGP lacks built-in authentication, \ie one AS can broadcast virtually arbitrary BGP announcements to disrupt the security and reliability of Internet routing.
BGP anomalies can be classified into two categories: hijacking and route leak, as illustrated in Fig.~\ref{fig:bgp-anomaly}. BGP hijacking itself has two subcategories: \first falsely claiming the ownership of a prefix or \second announcing fake paths (usually more preferable than real paths) to prefixes. The first type of hijacking is solvable by Route Origin Validation (ROV)~\cite{mohapatra2013bgp}, which is experiencing gradual deployment. Yet, the second type of hijacking usually needs per-hop path validation protocols such as BGPsec~\cite{lepinski2017bgpsec} that has very limited deployment.
Also, BGP hijacking can be classified as \emph{prefix hijacking} (targeting a prefix of others) or \emph{subprefix hijacking} (targeting a subset of others' prefix, \ie subprefix).

The other category of BGP anomalies is route leak: a misbehaved AS propagates BGP announcements to another AS in violation of the intended policies, resulting in traffic forwarded through unintended links.
The Gao-Rexford model~\cite{gao2001inferring} describes the restrictions on BGP route propagation and can be used to identify BGP route leak (\ie the valley-free criterion).
For example, in 2019, AS 21217 (\emph{Safe Host}) broke the valley-free criterion by propagating announcements received from its providers (\eg AS 13237 (\emph{euNetworks GmbH})) to another provider AS 4134 (\emph{China Telecom}), redirecting large amounts of Internet traffic destined for European mobile networks through \emph{China Telecom}~\cite{european2019}. 

%% file: content/sec3-model-method.tex
\section{Semantics Aware Analysis} 
\label{sec:method}

\begin{figure}[!ht]
    \centering
    \includegraphics[width=\linewidth]{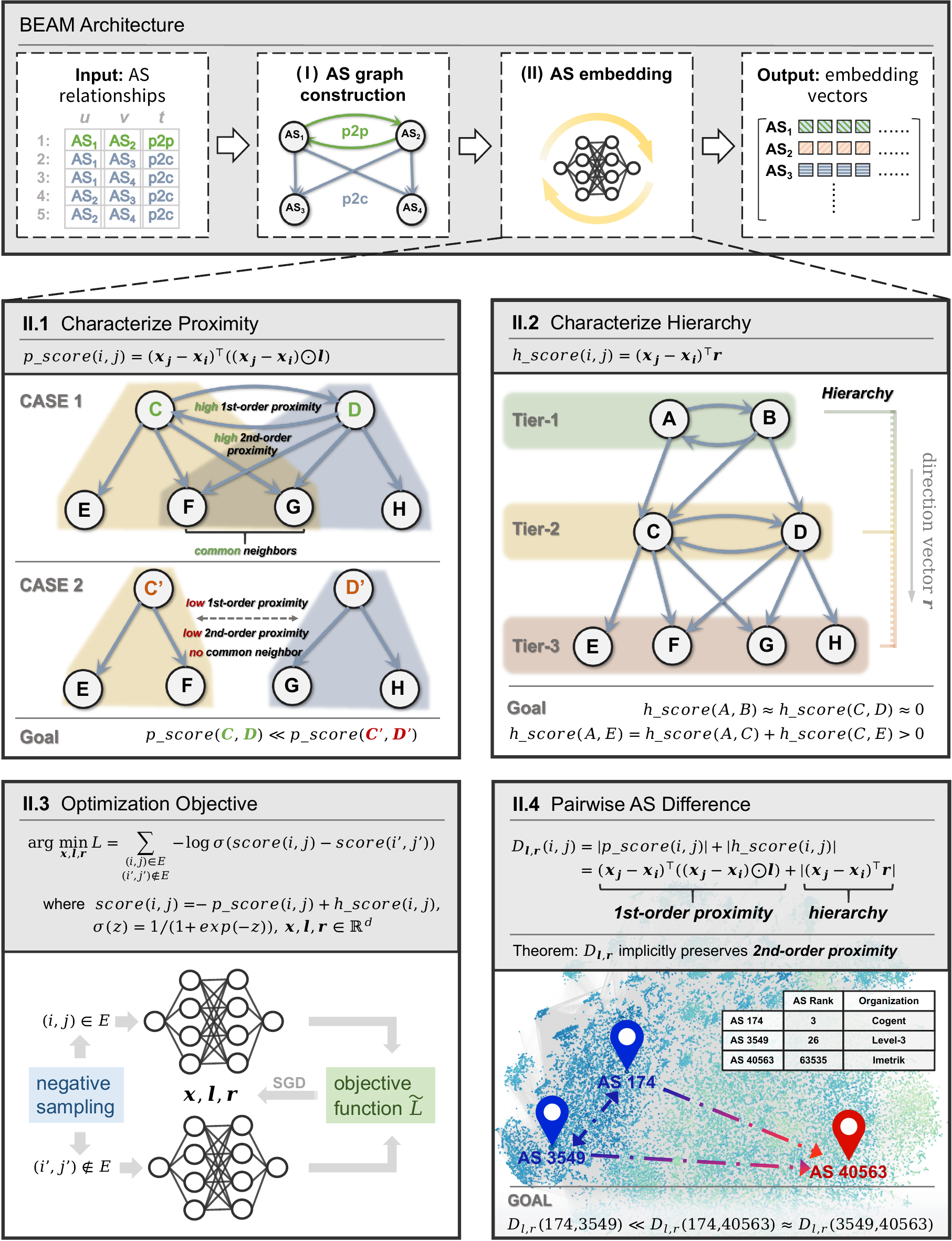}
    \caption{\textbf{Learning AS routing roles via \model.}
    \textnormal{\model takes AS relationships as the input and outputs the embedding vectors that represent AS routing roles.
    (\textbf{II.1})
    \model characterizes the proximity between ASes by $p\_score$. If two ASes are directly connected and have the same business relationships with many common neighbors, their proximity tends to be high, \ie a lower $p\_score$.
    (\textbf{II.2}) \model characterizes the hierarchy among ASes by $h\_score$. If an AS must traverse multiple consecutive P2C links to reach another AS, their hierarchy difference should be large, \ie a higher $h\_score$. 
    (\textbf{II.3}) \model utilizes one joint objective to \update{optimize} both $p\_score$ and $h\_score$;
    negative sampling is also applied.
    (\textbf{II.4}) The function $D_{\bm{l},\bm{r}}$ measures the routing-role difference between two ASes. A higher $D_{\bm{l},\bm{r}}$ value means higher difference. }}
    \label{fig:method}
\end{figure}

\subsection{\model Overview}
\label{sec:method:overview}

We propose a novel network representation learning model, \model, to learn the routing roles of ASes. The routing roles meaningfully characterize the ASes
in BGP route announcements and are utilized to detect Internet routing anomalies. 
As shown in \mbox{Fig.~\ref{fig:method}}, \model takes the AS business relationships as the input and generates the embedding vector for each AS by \first constructing an AS graph, and \second performing AS embedding. To compute AS embedding vectors, \model designs distance functions to measure two routing characteristics of the ASes: the proximity (see \mbox{Fig.~\ref{fig:method}(II.1)}) and the hierarchy (see \mbox{Fig.~\ref{fig:method}(II.2)}). The exact definitions of the two characteristics are given in \S\ref{sec:method:problem-definition}. We design a dedicated optimization objective such that \model preserves both the proximity and the hierarchy through the embedding (see \mbox{Fig.~\ref{fig:method}(II.3)}).
After obtaining the optimized embedding vectors, we can further measure the pairwise AS difference in terms of their routing roles (see \mbox{Fig.~\ref{fig:method}(II.4)}).

To our best knowledge, \model is the first dedicated network representation learning model that fully integrates BGP semantics into the training process and enables meaningful and accurate representation of ASes' routing roles.
\update{Applying network representation learning, rather than using ``raw'' AS business relationships, is essential to characterize ASes' routing roles. In particular, our network representation learning model can capture the global routing characteristics of each AS and translate them into embedding vectors, while the original AS business relationships can only indicate the local routing policy between two directly connected ASes. 
Further, with the embedding vectors, we can quantify the difference in routing roles between \emph{any} pair of ASes, regardless of whether they are connected or not. This enables us to detect the unexpected AS routing role churns in the global Internet that capture routing anomalies.
}

\subsection{Model Formulation}
\label{sec:method:problem-definition}

\begin{definition}[AS Graph]
    An AS graph is a directed graph $G=(V,E)$ where each vertex $v \in V$ represents an AS and each directed edge $e=(u,v) \in E$ represents a P2C relationship from $u$ to $v$.
\end{definition}

We regard each unique AS (identified by its ASN) as a vertex, and the P2C relationship between two vertices as a directed edge from provider to customer. Accordingly, a C2P relationship is viewed as a reversed P2C one, and a P2P relationship is represented via two edges in opposite directions.

We define two types of AS proximity
to represent the similarity between two ASes according to their local connections/relationships with their neighbors. 

\begin{definition}[First-Order AS Proximity]
    The first-order proximity between two ASes is their pairwise connection in the AS graph.
    \update{For a pair of vertices $(u, v)$, the first-order proximity between them is 1 if they are connected by an edge $e = (u,v) \in E$; otherwise it is 0.}
\label{def:first-order-pro}
\end{definition}

\begin{definition}[Second-Order AS Proximity]
     The second-order proximity between two ASes is
     the similarity between their neighborhood network structures. Given vertices $u, v$, let $\bm{p_u}=\langle \bm{w_{u,1}},\dots,\bm{w_{u,|V|}} \rangle$ denote the first-order proximity of AS $u$ with other ASes, \update{the second-order proximity between $u$ and $v$ is quantified based on the consistency between $\bm{p_u}$ and $\bm{p_v}$}. 
\label{def:second-order-pro}
\end{definition}

While the concept of proximity has been proposed before~\cite{tang2015line}, \model is the first to extend its interpretation to BGP semantics.
As illustrated in \mbox{Fig.~\ref{fig:method}(II.1)}, ASes C and D have high first-order proximity due to the P2P relationship (\ie two edges in the AS graph),
and also high second-order proximity since they provide Internet transit services for a similar set of customers. We also confirm our interpretation via a real-world example: AS 8903 and AS 12541, owned by
cloud service provider \emph{Evolutio}, serve as each other's backup and hence have very similar routing roles.
In terms of proximity, they have 25 common customers and their neighbor AS sets are highly similar (with a Jaccard index of 86.2\%); thus, the proximity is consistent with routing role similarity. 

As discussed in \S\ref{sec:background}, the Internet topology exhibits hierarchy.
Typically, a provider AS is considered on a higher level than its customers. Thus, we define AS hierarchy as follows:

\begin{definition}[AS Hierarchy]
    \update{The hierarchy of an AS is its tendency to establish P2C relationship with other ASes.}
    For two vertices $u, v$, if there exists a directed edge $e=(u,v)$, the hierarchy from $u$ to $v$ is positive.
\label{def:hierarchy}
\end{definition}

In real world, AS 7018 (\emph{AT\&T}) lies on the top of the Internet since it establishes either P2C or P2P relationship with other ASes,
while AS 140061 (\emph{China Telecom}) is a stub AS without any customers. Thus, AS 7018 and AS 140061 have very different routing roles in BGP, which is aligned with the positive hierarchy between them.

To quantify the AS proximity and hierarchy, we embed ASes into low-dimensional representations.

\begin{definition}[AS Embedding]
     Given $G=(V,E)$,
     AS embedding is to map each vertex $v \in V$ into a low-dimensional vector space $\mathbb{R}^d$, \ie learn a mapping function $f_{G;\theta}:V \rightarrow \mathbb{R}^d$, where $\theta$ contains learnable parameters and $d \ll |V|$.
\label{def:as-embedding}
\end{definition}

Finally,
we define our \model model as follows:

\begin{definition}[\modelfull]
     Given AS relationships, the \model model constructs the AS graph $G=(V,E)$ and performs AS embedding,
     such that the embedding vectors $\bm{x} = \{ \bm{x_v}|v \in V, \bm{x_v} = f_{G;\theta}(v)\}$ preserve the first- and second-order proximity and the hierarchy of ASes.
\end{definition}

\subsection{AS Graph Construction}
\label{sec:method:as-graph-construction}

The first step of training \model is to construct the AS graph. We use the real-world CAIDA AS relationship dataset~\cite{caida_as_relationship}\footnote{\update{CAIDA is not the only source of AS business relationships. Other sources like TopoScope~\cite{jin2020toposcope} are also available for training our system.
}} to construct the AS graph $G$. A business relationship between two ASes can be denoted as a tuple $(u, v, t)$, where $u$ and $v$ are two ASNs and $t \in \{\textsc{P2P}, \textsc{P2C}, \textsc{C2P}\}$ refers to the relationship type.
For each tuple $(u, v, t)$ in the CAIDA dataset, if $t=\textsc{P2C}$, we add a directed edge $e=(u,v)$ into $E$. If $t=\textsc{C2P}$, we add a directed edge $e=(v,u)$ into $E$. And if $t=\textsc{P2P}$, we add two directed edges $e=(u,v)$ and $e'=(v,u)$ into $E$.

The rationale for using AS business relationship to construct the AS graph is that it primarily determines how an AS chooses to update the routing paths received from neighbors,
and how the new generated route announcements are propagated~\cite{gao2001inferring}. 
Hence, it has direct impacts on ASes' routing roles. Moreover, the AS business relationship is a relatively stable property determined by real-world commercial agreements between connected ASes.
Thus, it is challenging to impersonate a specific AS without simultaneously changing or faking multiple AS relationships.
Besides, unlike historical route announcement data, the AS business relationship dataset contains fewer incorrect entries~\cite{luckie2013relationships}.

\subsection{AS Embedding}
\label{sec:method:as-embedding}

With a constructed AS graph, we embed ASes into a vector space while preserving proximity and hierarchy.
To this end, we design two distance functions to measure the difference between ASes regarding proximity and hierarchy, respectively.

\noindent\textbf{Proximity Distance.}
The proximity distance,
indicated by $p\_score$, models the first-order proximity between ASes: a small distance between two ASes means that their proximity is large. Per Def.~\ref{def:first-order-pro}, given a pair of directly linked vertices $(u, v)$
and another pair of vertices $(u', v')$ without a direct edge, $p\_score$ should
reflect the difference between their first-order proximity, \ie $p\_score(u, v) < p\_score(u', v')$. Therefore, we define $p\_score$ as follows:
\begin{equation}
    p\_score(u, v) = (\bm{x_v}-\bm{x_u})^{\intercal}\big((\bm{x_v}-\bm{x_u}) \odot \bm{l}\big), 
    \label{eq:s1}
\end{equation}
\noindent where $\bm{x_u}, \bm{x_v} \in \mathbb{R}^d$ denote the embedding vectors of $u, v$, respectively. $\bm{l} \in \mathbb{R}^d$ is a learnable weight vector for the $d$ components. \update{$\intercal$ and $\odot$ denote matrix transpose and Hadamard product, respectively.}
\update{Intuitively, $\bm{l}$ projects the embedding vectors into a subspace intended for preserving the proximity.}
To explicitly preserve the first-order proximity,
\model learns $\bm{x}, \bm{l}$ by decreasing the $p\_score$ of two vertices with an edge, while increasing the $p\_score$ of two vertices that are not directly connected.
Since this training strategy also preserves the second-order proximity implicitly (elaborated later in this section),
we do not define a dedicated distance function for the second-order proximity. 

\noindent\textbf{Hierarchy Distance.}
The hierarchy distance, indicated by $h\_score$, quantifies the hierarchical difference between ASes. Per Def.~\ref{def:hierarchy}, given a pair of vertices $(u, v)$ where $(u, v) \in E, (v, u) \notin E$, and another pair of vertices $(u',v')$ without a directed edge, \ie $(u',v') \notin E$, the $h\_score$ should reflect the difference between their hierarchy, \ie $h\_score(u, v) > h\_score(u', v')$. To this end, we design $h\_score$ as follows:
\begin{equation}
    h\_score(u, v) = (\bm{x_v}-\bm{x_u})^{\intercal}\bm{r}, 
    \label{eq:s2}
\end{equation}
\noindent where $\bm{x_u}, \bm{x_v} \in \mathbb{R}^d$ denote the embedding vectors of $u, v$, respectively.
$\bm{r} \in \mathbb{R}^d$ is a learnable unit vector indicating the descending direction of hierarchy.
\update{Intuitively, $\bm{r}$ projects the embedding vectors into a subspace intended for preserving the hierarchy, and thus}
the $h\_score$ calculates the projected length of $\bm{x_v}-\bm{x_u}$ on the specific direction vector $\bm{r}$ such that it has two important properties, \ie $h\_score(u, v) = -h\_score(v, u)$ and $h\_score(u, v) = h\_score(u, w)+h\_score(w, v)$.
To explicitly preserve the hierarchy of ASes,
\model learns $\bm{x}, \bm{r}$ by increasing the $h\_score$ of two vertices with a directed edge (\ie a P2C relationship), while decreasing $h\_score$ of two vertices not directly connected.
Since ASes with P2P relationship is connected by two edges in opposite directions, their $h\_score$ would approach zero under this training strategy, which is consistent with the BGP semantics that two peering ASes are typically on the same hierarchy of the Internet.

\noindent\textbf{Training Objective.}
\label{sec:method:training-objective}
With the two distance functions, we design the training objective of \model to ensure that a trained \model preserves both proximity and hierarchy.
We consolidate the two distance functions as follows:
\begin{equation}
  score(u, v) = -p\_score(u, v) + h\_score(u, v). 
  \label{eq:s}
\end{equation}
\update{Since a small $p\_score$ means large proximity and a large $h\_score$ means large hierarchy, we subtract $p\_score$ in Eq.\eqref{eq:s} so that $score$ increases monotonically with the difference between ASes in terms of proximity and hierarchy. This design allows BEAM to preserve the two routing characteristics better.
Although $score$ may become zero in some cases,
it will not affect \model’s training since the training objective is to enlarge the difference of $score$ between observed edges and nonexistent edges.}
Per Def.~\ref{def:first-order-pro} and~\ref{def:hierarchy}, given an observed edge $(u, v)$ and a nonexistent edge $(u', v')$, our model should assign $score(u, v) > score(u', v')$ with the optimal $\bm{x}$, $\bm{l}$ and $\bm{r}$. 
To this end, we formulate the optimization problem as follows:
\begin{equation}
  \arg \min_{\bm{x},\bm{l},\bm{r}} L
    = \sum_{\substack{(u,v) \in E \\
            (u',v') \notin E}} -\log\sigma(score(u,v)-score(u',v')), 
  \label{eq:objective}
\end{equation}
\noindent where $E$ is the edge set of the AS graph, $\sigma(z)=\frac{1}{1+exp(-z)}$ is the sigmoid function, and $L$ is the objective function to be minimized. 
\update{We solve this problem via a fully connected neural network, which has an embedding layer and two linear layers. The embedding layer generates the embedding vectors ($\bm{x}$) that represent ASes' routing roles and the two linear layers project the embedding vectors into two subspaces that preserve proximity ($\bm{l}$) and hierarchy ($\bm{r}$), respectively. For each edge $(u, v) \in E$, we sample 10 negative (nonexistent) edges $(u', v') \notin E$ and each $((u, v),(u', v'))$ forms one training instance. The neural network generates the embedding vectors of $u$ and $v$, 
computes the loss via Eq.\eqref{eq:objective}, and uses SGD~\cite{robbins1951stochastic} to optimize itself. When the training is complete, the neural network learns the ASes' routing roles. We empirically set $d = 128$ and train the network for 1,000 epochs. The batch size is 1024 and the initial learning rate is $10^{-5}$.
}



\noindent\textbf{Computing Pairwise AS Difference.}
\label{sec:method:as-difference-measurement}
The pairwise AS difference represents their difference in the routing roles,
which we define
between the embedding vectors of two ASes using the \model model (including its parameters $\bm{x}$, $\bm{l}$, and $\bm{r}$):
\begin{equation}
\begin{split}
D_{\bm{l},\bm{r}}(u, v) &= |p\_score(u, v)|+|h\_score(u, v)| \\
        = &\underbrace{(\bm{x_v}-\bm{x_u})^{\intercal}((\bm{x_v}-\bm{x_u})\odot \bm{l})}_\text{\emph{the first-order proximity}}+\underbrace{|(\bm{x_v}-\bm{x_u})^{\intercal}\bm{r}|}_\text{\emph{the hierarchy}}.
\end{split}
\label{eq:distance}
\end{equation}

\noindent Note that this pairwise AS difference directly reflects
the first-order proximity and the hierarchy between ASes. 
Moreover, per the definition of the second-order proximity (Def.~\ref{def:second-order-pro}), the pairwise AS difference between two vertices should be small if their neighbors and the business relationships with their neighbors are similar.
In Appendix~\ref{sec:appendix:proof-of-the-second-order-proximity-preservation}, we prove a theorem that this pairwise AS difference does preserve the 
second-order proximity between ASes.
Thus, our \model model can preserve the first-order proximity, the second-order proximity and the hierarchy between ASes.

%% file: content/sec3-model-result.tex
\subsection{Embedding Results Analysis}
\label{sec:method:embedding-results}

We train \model with the CAIDA AS relationship dataset collected on 06/01/2018 to study the routing roles of ASes. The dataset contains 61,549 ASes and 439,981 business relationships, and is randomly selected from the CAIDA datasets collected before historical BGP incidents (see \S\ref{sec:measurement:analysis-results}).


\begin{figure*}[ht]
    \centering
    \includegraphics[width=0.8\linewidth] {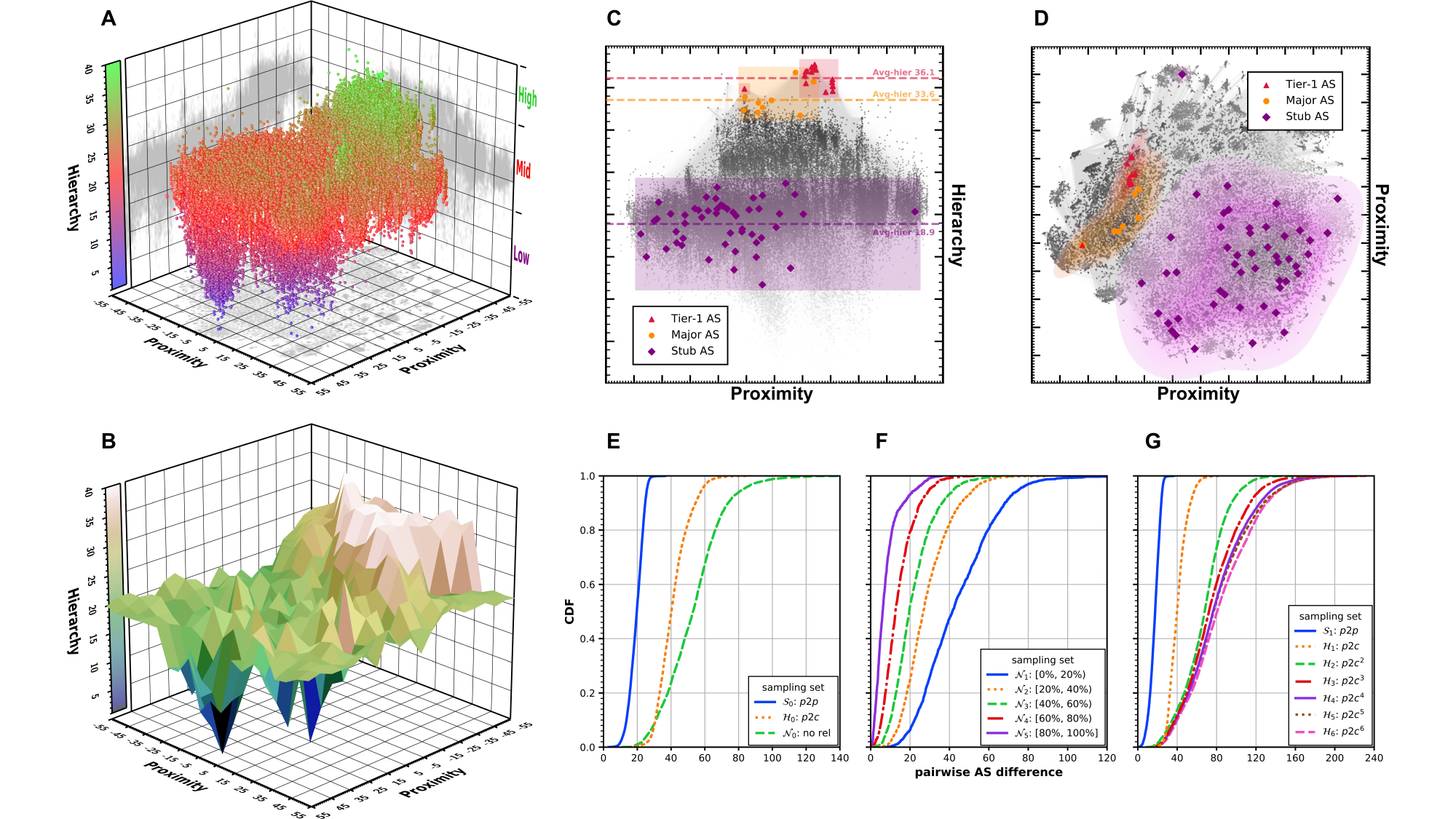}
    \caption{\label{fig:emb-result}\textbf{Embedding results.}
    \textnormal{(\textbf{A}) All embedding vectors visualized in a 3-D space.
    The Z-axis represents the hierarchy level and the XY-plane reflects the proximity.
    (\textbf{B}) The terrain plot of the
    embedding vectors. The estimated spatial distribution shows the overall characteristics of
    AS routing roles.
    (\textbf{C}) The YZ-plane projection of the
    embedding vectors. The X-axis and the Y-axis of the projection plane indicate the proximity and the hierarchy, respectively. The areas with different colors illustrate the distribution of typical ASes (\ie vertices with different colors), and the dashed lines show their average hierarchy levels.
    (\textbf{D}) The XY-plane projection of the
    embedding vectors.
    (\textbf{E}) The routing role difference regarding the sampled subset $S_0$, $H_0$ and $N_0$.
    (\textbf{F}) The routing role difference regarding the sampled subset $N_1$ to $N_5$. 
    (\textbf{G}) The routing role difference regarding
    $S_1$ and $H_1$ to $H_6$.
    }}
\end{figure*}

\noindent\textbf{Visualizing Embedding Vectors.}
We visualize the embedding vectors computed by \model in a 3-D space to illustrate the overall characteristics of AS routing roles.
Since \model learns
the unit direction vector $\bm{r}$ that represents the descending direction of hierarchy, we decompose each embedding vector into two parts: the projection on $\bm{r}$,
and the projection on the plane orthogonal to $\bm{r}$ (\ie the rejection).
We use the length of the projection as the coordinate value of the Z-axis to represent the hierarchy level of each AS.
We further transform the rejection into a 2-D space by the widely used dimension reduction method t-SNE~\cite{van2008visualizing}, and hence obtain the coordinate values of the X-axis and the Y-axis.

We visualize all embedding vectors in \mbox{Fig.~\ref{fig:emb-result}(A)}, where each vertex represents a unique
AS and the color
indicates the AS hierarchy level. 
We observe
that a few vertices are densely located (\ie high proximity) on the highest and the lowest levels of the Internet, while more medium-level vertices are sparsely distributed (\ie low proximity). 
For a better illustration,
we show the terrain plot of the same 3-D space in \mbox{Fig.~\ref{fig:emb-result}(B)} via
IDW~\cite{lu2008adaptive},
an interpolation metric widely used for
spatial distribution estimation.
The observations are consistent with the BGP fact that a small number of ASes on the highest hierarchy level have similar routing roles, since they all provide transit services for other ASes and establish P2P relationships with the ASes on the same level to form ``Internet backbone''. On the contrary, it takes multiple C2P links for another small set of ASes to reach the backbone ASes. Thus, they lie on the lowest hierarchy level, forming several clusters with lower Z-axis coordinate values. 
The rest ASes are located on the medium hierarchy. They have different providers and customers, and serve diverse routing roles. 

\noindent\textbf{Routing Role Analysis.}
We further study the routing roles of
specific ASes to 
check if the computed embedding vectors reflect their actual routing properties in the Internet.
In particular, we choose 16 Tier-1 ASes (\eg \emph{AT\&T}), 9 major yet not Tier-1 ASes (\eg \emph{Telstra}), and 50 other 
random
stub ASes (\ie the ASes connected to only one other AS). We present the projections of their routing roles on both YZ-plane (\mbox{Fig.~\ref{fig:emb-result}(C)}) and XY-plane (\mbox{Fig.~\ref{fig:emb-result}(D)}). The results
show
that the Tier-1 ASes form dense clusters on the highest hierarchy level, the major yet not Tier-1 ASes form several dense clusters on the levels relatively lower than those of Tier-1 ASes, and the
stub ASes are dispersed on the much lower levels. 
This confirms that the computed embedding vectors preserve both the proximity and hierarchy of ASes.
\mbox{Figure~\ref{fig:emb-result}(C)}, however, does not exhibit 
distinct hierarchies, which is expected since
the Internet topology is not strictly hierarchical~\cite{luckie2013relationships};
for example, AS 4134 (\emph{China Telecom}) has P2P relationship with the Tier-1 AS 1299 (\emph{Arelion}), yet it also keeps a 
route with 2 consecutive C2P links to reach AS 1299.

\noindent\textbf{Pairwise AS Difference Comparison.}
We measure the difference of routing roles between two ASes. 
We choose AS pairs in three categories: \first We randomly sample a subset $S_0$ from all pairs of ASes with the P2P relationship while enforcing the following constraint. In particular, given an AS pair, we obtain the neighbor AS set of each AS in the pair. The AS pair is eligible only if the Jaccard index (\ie the similarity) of the two neighbor AS sets is between 0\% and 10\%. We also sample another subset $S_1$ without enforcing this constraint. 
\second We randomly sample six subsets (denoted as $N_0$ to $N_5$) from all pairs of ASes that are not connected directly, where the Jaccard indices of the two neighbor AS sets are between 0\% and 10\%, 0\% and 20\%, 20\% and 40\%, 40\% and 60\%, 60\% and 80\%, and 80\% and 100\%, respectively.
\third We randomly sample seven subsets (denoted as $H_0$ to $H_6$), where the shortest path between two ASes in one pair is a direct P2P link, and 2, 3, 4, 5, and 6 consecutive P2C links, respectively.
For instance, if two AS pairs (\eg AS1 and AS2, AS2 and AS3) are both in the P2C relationship, then AS1 and AS3 form two consecutive P2C links. 
See details in Appendix~\ref{sec:appendix:as-pair-datasets}.

For each sampled subset, we compute the difference of routing roles between two ASes 
and obtain Cumulative Distribution Function (CDF) curves. We make three types of comparisons among the CDF curves of different subsets.
The comparisons of the P2P AS pairs ($S_0$), P2C AS pairs ($H_0$) and no-relationship AS pairs ($N_0$) show that $S_0$ has the smallest difference and $N_0$ exhibits the largest difference (see \mbox{Fig.~\ref{fig:emb-result}(E)}).
It reveals that \model preserves the first-order proximity. The comparisons among the subsets with different degrees of neighbor intersection ($N_1$ to $N_5$) show that if two ASes have more common neighbors (\ie high second-order proximity), they tend to have similar routing roles (see \mbox{Fig.~\ref{fig:emb-result}(F)}). In  \mbox{Fig.~\ref{fig:emb-result}(G)}, the difference of routing roles increases with the number of consecutive P2C links between two ASes, which is consistent with the nature of AS hierarchy.

%% file: content/sec4-detection.tex
\section{The Anomaly Detection System}
\label{sec:detection}

In this section, we develop a semantics aware routing anomaly detection system built upon \model\footnote{\update{Our system is open source at  \href{https://github.com/anonymized-for-reviews/routing-anomaly-detection}{this GitHub repository}}.}.

\begin{figure*}[ht]
    \centering
    \includegraphics[width=0.9\linewidth]{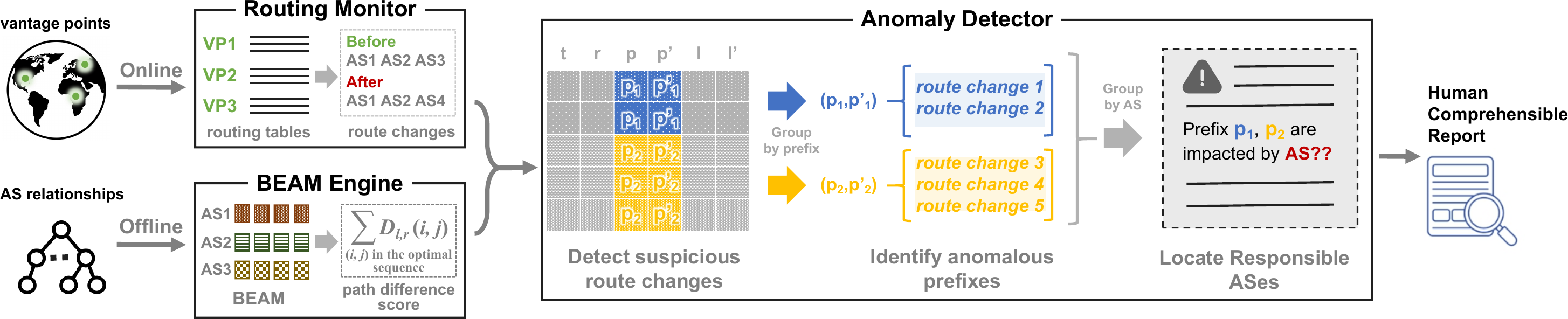}
    \caption{\textbf{The workflow of our routing anomaly detection system built upon \model.}}
    \label{fig:detection-system}
\end{figure*}

\subsection{System Overview}
\label{sec:detection:overview}

Our routing anomaly detection system is designed to detect global-scale Internet routing anomalies.
As shown in Fig.~\ref{fig:detection-system}, our detection system consists of three components: the routing monitor, the \model engine, and the anomaly detector. The routing monitor establishes connections with global vantage points to capture route changes in real time. The \model engine utilizes our \model model to compute the \emph{path difference scores} of the route changes, which quantifies the routing role difference between the new and the original routing paths. Based on the path difference scores, the anomaly detector identifies  suspicious route changes, and groups the suspicious route changes sharing the same prefix into anomalous prefix events. This is necessary to identify the prefixes impacted by widespread routing anomalies. 
It further locates the ASes responsible for each anomalous prefix event and correlates the events caused by the same set of responsible ASes. 
Since the \model model used in our detection system is pretrained with AS relationship data (instead of labeled routing anomaly data), our detection is unsupervised.

\subsection{BEAM Engine}
\label{sec:measurement:path-differnce-score}



Not all route changes are caused by routing anomalies. For instance, different ASes of the same organization may claim the ownership of a specific prefix simultaneously, thus generating two routing paths with different origin ASes. This
is called multiple origin AS (MOAS)~\cite{zhao2001analysis} and should not be considered as anomalous. Thus, the BEAM engine relies on \emph{path difference score} to identify suspicious route changes.  


To compute the path difference score for a route change, we design a method based on the dynamic time warping (DTW) algorithm~\cite{berndt1994using}, an effective way of measuring the overall difference between two ordered sequences of unequal length. 
Specifically, given two routing paths $S=\langle v_1,\dots, v_m \rangle$ and $S'=\langle v'_1,\dots, v'_n \rangle$, the DTW algorithm repeatedly selects a pair of ASes from $S$ and $S'$, and generates 
eligible sequences of AS pairs by satisfying the following conditions: \first Each AS in $S$ ($S'$) should be paired with one or more ASes from $S'$ ($S$). \second The first (last) AS in $S$ should be paired with the first (last) AS in $S'$. \third Given $i<j$, if $v_i$ and $v_j$ from $S$ are paired with $v'_k$ and $v'_l$ in $S'$, then there keeps $k \le l$; S and S' are symmetric. 
\update{For example, if $S=\langle v_1,v_2,v_3 \rangle$, $S'=\langle v'_1,v'_2,v'_3 \rangle$, then $((v_1,v'_1),(v_2,v'_2),(v_3,v'_3))$ and $((v_1,v'_1),(v_1,v'_2),(v_2,v'_2),(v_3,v'_3))$ are both eligible, but $((v_1,v'_1),(v_1,v'_2),(v_2,v'_1),(v_3,v'_3))$ is ineligible because the two pairs $(v_1,v'_2)$ and $(v_2,v'_1)$ violate the condition \third.}
For each eligible sequence, we sum the pairwise AS difference (given by \model) of its AS pairs. Then, we choose the minimum value from all summed values as the path difference score. The sequence of AS pairs with the minimum value is referred to as \emph{the optimal sequence}. Let $S[1:i]$ ($S'[1:j]$) denote the first $i$ ($j$) ASes in $S$ ($S'$). We apply dynamic programming to obtain the optimal sequence of AS pairs for $S[1:i]$ and $S'[1:j]$ based on prior states as follows:
\begin{itemize}[itemsep=0em,align=parleft,left=0pt..1em]
    \item Given the optimal sequence for $S[1:i]$ and $S'[1:j-1]$, add a new AS pair $(S[i], S'[j])$ to the sequence.
    \item Given the optimal sequence for $S[1:i-1]$ and $S'[1:j]$, add a new AS pair $(S[i], S'[j])$ to the sequence.
    \item Given the optimal sequence for $S[1:i-1]$ and $S'[1:j-1]$, add a new AS pair $(S[i], S'[j])$ to the sequence.
\end{itemize}

\noindent
We choose the operation yielding the minimal sum of pairwise AS difference as the optimal AS pair sequence for $S[1:i]$ and $S'[1:j]$.
The induction base, \ie the optimal AS pair sequence for $S[1:1]$ and $S'[1:1]$, is trivial to compute. The pseudo-code of this algorithm is presented in Appendix~\ref{sec:appendix:detection-system}.

\subsection{Anomaly Detector}
\label{sec:detection:procedure}


\noindent\textbf{Detecting Suspicious Route Changes.}
The anomaly detector first identifies suspicious route changes caused by routing anomalies. Specifically, for a route change, 
the anomaly detector checks whether its path difference score is greater than a threshold $th_d$, which is dynamically computed using historical legitimate route changes (detailed in \S\ref{sec:detection:results}). If so, the route change is regarded as \emph{suspicious}.

\noindent\textbf{Identifying Anomalous Prefixes.}
It is necessary to prioritize the widespread routing anomalies captured by multiple vantage points. Towards this end, the anomaly detector groups the suspicious route changes that impact the same prefix into different \emph{prefix events}, where each event is associated with a specific prefix and sorts the suspicious route changes by their occurrence time. For each event, the anomaly detector applies a sliding window to count the number of individual vantage points that observe suspicious route changes within the window. If this number is above a threshold $th_v$ (detailed in \S\ref{sec:detection:results}), we consider the prefix event is associated with a widespread routing anomaly and regard the event as anomalous.


\noindent\textbf{Locating Responsible ASes.}
The misbehaved ASes that are responsible for routing anomalies 
may impact multiple prefixes simultaneously. 
To obtain comprehensive information about the affected prefixes in each routing anomaly, the anomaly detector correlates all anomalous prefix events based on their responsible ASes. In particular, for each suspicious route change associated with an anomalous prefix event, our detector identifies the ASes that either appear in the new path or in the original path. Then, we compute the intersection of these ASes from all route changes as the responsible ASes for the anomalous prefix event. Given two prefix events, if their time ranges have overlaps and they have common responsible ASes, we consider they are correlated. Thereby, we can divide all anomalous prefix events into multiple sets, where each event only correlates with the other events in the same set (\ie no cross-set correlation). Finally, our anomaly detector treats each set as an individual routing anomaly and outputs an  alarm that specifies both the affected prefixes and the responsible ASes. The additional details of our detection system is supplemented in Appendix~\ref{sec:appendix:detection-system}.

%% file: content/sec5-detection-result.tex
\section{Experimental Results} 
\label{sec:measurement:analysis-results}

In this section, we evaluate the path difference scores and perform experiments with real-world BGP data. 
We also deploy our system at a large ISP to verify its effectiveness in practice.

\subsection{Measuring Path Difference Score} 
\label{subsec:path_diff_score_results}

\begin{figure*}[t]
    \centering
    \includegraphics[width=0.9\linewidth]{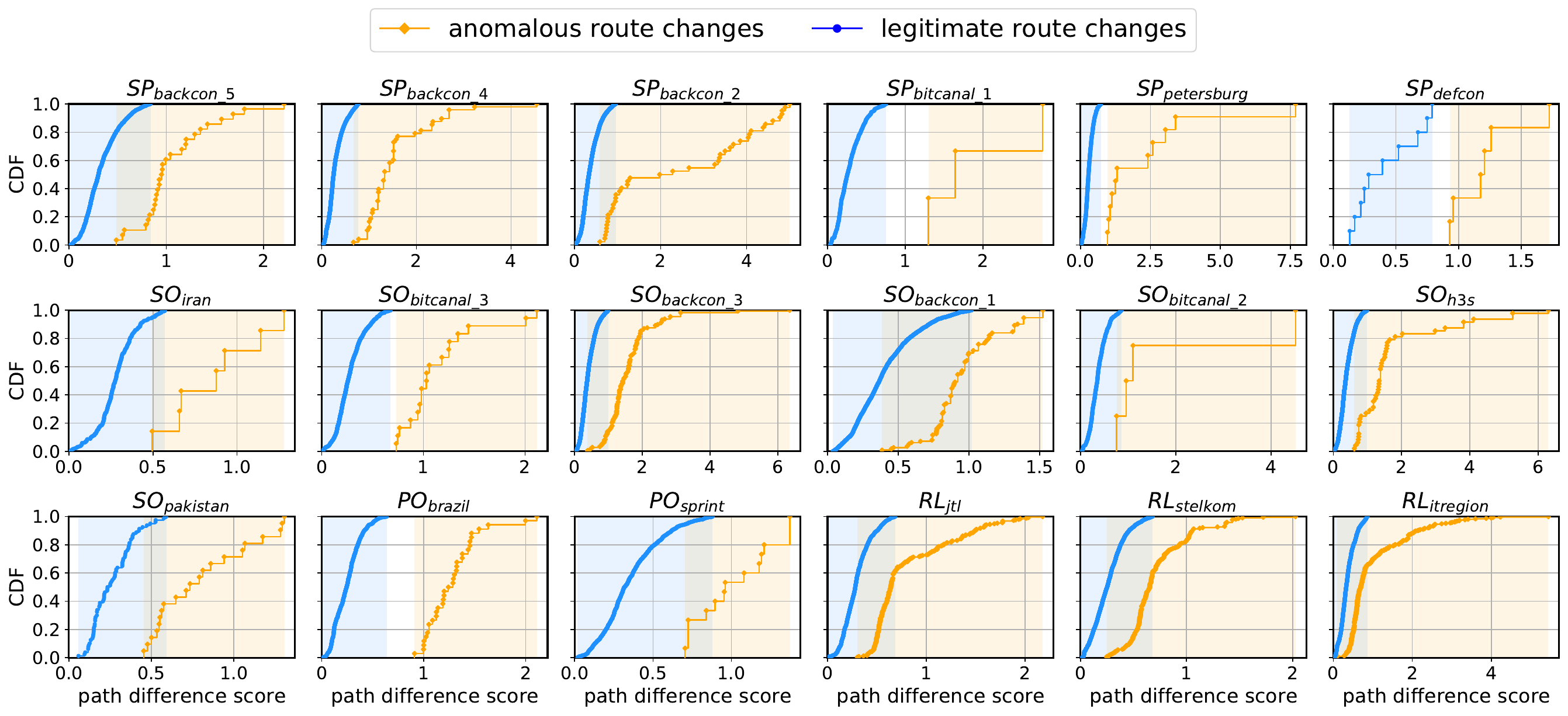}
    \caption{\textbf{Statistical comparisons of path difference scores between anomalous and legitimate route changes.}
    }
    \label{fig:mrc-brc}
\end{figure*}

We use real-world route announcements to analyze both legitimate and anomalous route changes in terms of path difference scores. 
We collect 18 reports on historical routing anomalies spanning from 2008 to 2021, including 15 BGP hijacking (2 prefix and 13 subprefix hijacking) and 3 BGP route leak incidents. 
For each anomaly, we obtain manually confirmed information (\eg the time of anomalies and the affected prefixes) from two authoritative sources, the Oracle blogs~\cite{oracle_blogs} and the BGPStream monitor~\cite{bgpstream}. Based on the information, we fetch all route announcements 12 hours before and after the anomalies from RouteViews~\cite{routeviews} and obtain 18 datasets. The total number of collected route announcements is 11,861,377,951. 
In each dataset, we identify the anomalous route changes
with
the confirmed information of the anomaly. The anomalous route changes in our datasets cover both \emph{origin change} (\ie two routing paths have different origin ASes) and \emph{path change} (\ie two routing paths share the same origin yet traverse different ASes).
We also locate legitimate route changes such as origin changes incurred by the multiple origin AS issue~\cite{zhao2001analysis}.
The details are in Appendix~\ref{sec:appendix:route-datasets} and~\ref{sec:appendix:legitimate-route-change}.

\mbox{Figure~\ref{fig:mrc-brc}} shows the path difference scores measured by the \model engine on the 18 datasets. The path difference scores of anomalous route changes are much higher than those of legitimate ones, indicating that routing anomalies would significantly change the routing roles of ASes on the routing paths. 
For instance, the anomalous route changes in $SO_{pakistan}$ changed the origin AS from AS 36561 (\emph{YouTube}) to AS 17557 (\emph{Pakistan Telecom}).
These two origin ASes are significantly different in their neighbors, geographic locations and hierarchical levels in the Internet.
In contrast, the different origin ASes in a legitimate route change often belong to the same organization and their neighbors are more similar (\eg the same upstream AS), resulting in much smaller routing role churn between the original and new routes. 

\subsection{Routing Anomaly Detection Results}
\label{sec:detection:results}

\begin{table*}[ht]
    \caption{\textbf{Detection results on the 18 real-world datasets.} \textnormal{The \cmark/\xmark\xspace indicates whether the confirmed anomaly of a dataset is \emph{detected}. If \emph{detected}, the number of all raised alarms (\#Alarms) and that of false alarms (\#FalseAlarms) are presented in bold.}}
    \label{tab:detection-result}
    \small
    \centering
    \begin{tabu}{X[2.5,l]|X[0.8,c]|X[0.8,c]|X[0.8,c]|X[0.8,c]|X[0.8,c]|X[0.8,c]|X[0.8,c]|X[0.8,c]|X[0.8,c]|X[1.2,c]|X[1.2,c]|X[1.2,c]|X[1.2,c]|X[1.2,c]|X[1.2,c]|X[1.2,c]|X[1.2,c]|X[1.2,c]}
    \hline
    \multirow{2}{*}{\textbf{Dataset}}&\multicolumn{9}{c|}{\textbf{Detected}}&\multicolumn{9}{c}{\textbf{\#Alarms(\#FalseAlarms)}}\\
\cline{2-19}    &{\footnotesize \textbf{ED}}&{\footnotesize \textbf{JI}}&{\footnotesize \textbf{Li}}&{\footnotesize \textbf{Ma}}&{\footnotesize \textbf{NV}}&{\footnotesize \textbf{SD}}&{\footnotesize \textbf{LS}}&{\footnotesize \textbf{AV}}&{\footnotesize \textbf{Ours}}&{\footnotesize \textbf{ED}}&{\footnotesize \textbf{JI}}&{\footnotesize \textbf{Li}}&{\footnotesize \textbf{Ma}}&{\footnotesize \textbf{NV}}&{\footnotesize \textbf{SD}}&{\footnotesize \textbf{LS}}&{\footnotesize \textbf{AV}}&{\footnotesize \textbf{Ours}} \\
    \hline
    {\footnotesize ${SP}_{backcon\_5}$}&\cmark&\cmark&\cmark&\cmark&\cmark&\xmark&\xmark&\xmark&\cmark&{\scriptsize \textbf{18(5)}}&{\scriptsize \textbf{12(3)}}&{\scriptsize \textbf{21(4)}}&{\scriptsize \textbf{15(3)}}&{\scriptsize \textbf{17(2)}}&{\scriptsize 9(1)}&{\scriptsize 62(31)}&{\scriptsize 5(1)}&{\scriptsize \textbf{34(2)}}\\
    {\footnotesize ${SP}_{backcon\_4}$}&\cmark&\cmark&\cmark&\cmark&\cmark&\xmark&\cmark&\cmark&\cmark&{\scriptsize \textbf{14(1)}}&{\scriptsize \textbf{8(1)}}&{\scriptsize \textbf{15(2)}}&{\scriptsize \textbf{13(1)}}&{\scriptsize \textbf{12(1)}}&{\scriptsize 7(1)}&{\scriptsize \textbf{42(17)}}&{\scriptsize \textbf{22(6)}}&{\scriptsize \textbf{21(0)}}\\
    {\footnotesize ${SP}_{backcon\_2}$}&\cmark&\cmark&\cmark&\cmark&\cmark&\xmark&\cmark&\xmark&\cmark&{\scriptsize \textbf{29(7)}}&{\scriptsize \textbf{21(7)}}&{\scriptsize \textbf{24(5)}}&{\scriptsize \textbf{25(7)}}&{\scriptsize \textbf{21(4)}}&{\scriptsize 8(3)}&{\scriptsize \textbf{38(13)}}&{\scriptsize 23(18)}&{\scriptsize \textbf{37(1)}}\\
    {\footnotesize ${SP}_{bitcanal\_1}$}&\xmark&\xmark&\xmark&\cmark&\cmark&\xmark&\xmark&\xmark&\cmark&{\scriptsize 16(0)}&{\scriptsize 16(0)}&{\scriptsize 14(0)}&{\scriptsize \textbf{18(0)}}&{\scriptsize \textbf{17(0)}}&{\scriptsize 7(0)}&{\scriptsize 67(36)}&{\scriptsize 30(10)}&{\scriptsize \textbf{16(0)}}\\
    {\footnotesize ${SP}_{petersburg}$}&\cmark&\cmark&\cmark&\cmark&\cmark&\xmark&\cmark&\cmark&\cmark&{\scriptsize \textbf{22(3)}}&{\scriptsize \textbf{14(0)}}&{\scriptsize \textbf{21(3)}}&{\scriptsize \textbf{20(1)}}&{\scriptsize \textbf{16(0)}}&{\scriptsize 12(2)}&{\scriptsize \textbf{66(28)}}&{\scriptsize \textbf{37(16)}}&{\scriptsize \textbf{24(0)}}\\
    {\footnotesize ${SP}_{defcon}$}&\cmark&\cmark&\cmark&\cmark&\cmark&\xmark&\cmark&\cmark&\cmark&{\scriptsize \textbf{7(2)}}&{\scriptsize \textbf{7(2)}}&{\scriptsize \textbf{9(3)}}&{\scriptsize \textbf{9(3)}}&{\scriptsize \textbf{9(3)}}&{\scriptsize 2(2)}&{\scriptsize \textbf{28(10)}}&{\scriptsize \textbf{17(9)}}&{\scriptsize \textbf{7(1)}}\\
    {\footnotesize ${SO}_{iran}$}&\cmark&\cmark&\cmark&\cmark&\cmark&\xmark&\xmark&\xmark&\cmark&{\scriptsize \textbf{15(1)}}&{\scriptsize \textbf{8(1)}}&{\scriptsize \textbf{24(5)}}&{\scriptsize \textbf{12(2)}}&{\scriptsize \textbf{16(3)}}&{\scriptsize 0(0)}&{\scriptsize 21(11)}&{\scriptsize 19(10)}&{\scriptsize \textbf{31(2)}}\\
    {\footnotesize ${SO}_{bitcanal\_3}$}&\xmark&\xmark&\xmark&\xmark&\cmark&\xmark&\xmark&\xmark&\cmark&{\scriptsize 26(4)}&{\scriptsize 24(3)}&{\scriptsize 29(6)}&{\scriptsize 25(3)}&{\scriptsize \textbf{26(5)}}&{\scriptsize 7(0)}&{\scriptsize 44(19)}&{\scriptsize 17(8)}&{\scriptsize \textbf{40(1)}}\\
    {\footnotesize ${SO}_{backcon\_3}$}&\cmark&\cmark&\cmark&\cmark&\cmark&\xmark&\xmark&\xmark&\cmark&{\scriptsize \textbf{32(8)}}&{\scriptsize \textbf{23(4)}}&{\scriptsize \textbf{27(6)}}&{\scriptsize \textbf{34(8)}}&{\scriptsize \textbf{34(9)}}&{\scriptsize 6(1)}&{\scriptsize 49(27)}&{\scriptsize 19(9)}&{\scriptsize \textbf{35(5)}}\\
    {\footnotesize ${SO}_{backcon\_1}$}&\cmark&\xmark&\cmark&\cmark&\cmark&\xmark&\xmark&\xmark&\cmark&{\scriptsize \textbf{19(6)}}&{\scriptsize 16(4)}&{\scriptsize \textbf{35(14)}}&{\scriptsize \textbf{18(4)}}&{\scriptsize \textbf{17(7)}}&{\scriptsize 0(0)}&{\scriptsize 63(35)}&{\scriptsize 25(11)}&{\scriptsize \textbf{18(3)}}\\
    {\footnotesize ${SO}_{bitcanal\_2}$}&\xmark&\cmark&\cmark&\cmark&\cmark&\xmark&\xmark&\xmark&\cmark&{\scriptsize 16(1)}&{\scriptsize \textbf{15(1)}}&{\scriptsize \textbf{17(2)}}&{\scriptsize \textbf{15(1)}}&{\scriptsize \textbf{16(1)}}&{\scriptsize 12(2)}&{\scriptsize 39(14)}&{\scriptsize 29(8)}&{\scriptsize \textbf{24(0)}}\\
    {\footnotesize ${SO}_{h3s}$}&\cmark&\xmark&\cmark&\cmark&\cmark&\xmark&\cmark&\xmark&\cmark&{\scriptsize \textbf{11(1)}}&{\scriptsize 3(0)}&{\scriptsize \textbf{15(3)}}&{\scriptsize \textbf{12(2)}}&{\scriptsize \textbf{9(0)}}&{\scriptsize 5(1)}&{\scriptsize \textbf{38(22)}}&{\scriptsize 27(8)}&{\scriptsize \textbf{14(0)}}\\
    {\footnotesize ${SO}_{pakistan}$}&\cmark&\cmark&\cmark&\cmark&\cmark&\xmark&\xmark&\xmark&\cmark&{\scriptsize \textbf{12(4)}}&{\scriptsize \textbf{8(2)}}&{\scriptsize \textbf{9(2)}}&{\scriptsize \textbf{10(4)}}&{\scriptsize \textbf{8(1)}}&{\scriptsize 1(0)}&{\scriptsize 26(14)}&{\scriptsize 2(0)}&{\scriptsize \textbf{10(1)}}\\
    {\footnotesize ${PO}_{brazil}$}&\cmark&\cmark&\cmark&\cmark&\cmark&\xmark&\xmark&\cmark&\cmark&{\scriptsize \textbf{30(5)}}&{\scriptsize \textbf{32(4)}}&{\scriptsize \textbf{30(5)}}&{\scriptsize \textbf{37(5)}}&{\scriptsize \textbf{25(3)}}&{\scriptsize 11(2)}&{\scriptsize 52(25)}&{\scriptsize \textbf{28(11)}}&{\scriptsize \textbf{51(1)}}\\
    {\footnotesize ${PO}_{sprint}$}&\xmark&\xmark&\xmark&\xmark&\xmark&\xmark&\cmark&\xmark&\cmark&{\scriptsize 20(0)}&{\scriptsize 18(0)}&{\scriptsize 16(0)}&{\scriptsize 22(2)}&{\scriptsize 19(2)}&{\scriptsize 10(2)}&{\scriptsize \textbf{84(24)}}&{\scriptsize 33(8)}&{\scriptsize \textbf{29(0)}}\\
    {\footnotesize ${RL}_{jtl}$}&\xmark&\xmark&\xmark&\xmark&\xmark&\xmark&\cmark&\xmark&\cmark&{\scriptsize 17(1)}&{\scriptsize 16(1)}&{\scriptsize 21(2)}&{\scriptsize 16(2)}&{\scriptsize 17(1)}&{\scriptsize 6(3)}&{\scriptsize \textbf{60(40)}}&{\scriptsize 21(11)}&{\scriptsize \textbf{46(5)}}\\
    {\footnotesize ${RL}_{stelkom}$}&\xmark&\cmark&\xmark&\xmark&\xmark&\xmark&\cmark&\xmark&\cmark&{\scriptsize 25(4)}&{\scriptsize \textbf{21(2)}}&{\scriptsize 34(6)}&{\scriptsize 26(3)}&{\scriptsize 23(3)}&{\scriptsize 4(0)}&{\tiny \textbf{284(225)}}&{\scriptsize 17(8)}&{\scriptsize \textbf{43(3)}}\\
    {\footnotesize ${RL}_{itregion}$}&\xmark&\xmark&\xmark&\xmark&\xmark&\xmark&\cmark&\xmark&\cmark&{\scriptsize 25(0)}&{\scriptsize 21(1)}&{\scriptsize 23(0)}&{\scriptsize 20(0)}&{\scriptsize 21(3)}&{\scriptsize 6(2)}&{\scriptsize \textbf{74(44)}}&{\scriptsize 23(9)}&{\scriptsize \textbf{44(4)}}\\
    \hline
    \textbf{Overall}&{\scriptsize 11/18}&{\scriptsize 11/18}&{\scriptsize 12/18}&{\scriptsize 13/18}&{\scriptsize 14/18}&{\scriptsize 0/18}&{\scriptsize 9/18}&{\scriptsize 4/18}&{\scriptsize 18/18}&{\scriptsize 354(53)}&{\scriptsize 283(36)}&{\scriptsize 384(68)}&{\scriptsize 347(51)}&{\scriptsize 323(48)}&{\scriptsize 113(22)}&{\tiny 1137(635)}&{\tiny 394(161)}&{\scriptsize 524(29)}\\
    \hline
    \end{tabu}
\end{table*}

We now validate the performance of our detection system. 
Since
RouteViews
archives the RIB data (\ie the snapshot of routing table) of global vantage points bi-hourly, for each dataset,
we fetch the most recent RIB data before the confirmed anomaly 
to initialize the
routing tables that our detection system monitors. Then we use the route announcements observed within two hours as input, which 
include all routes associated with the confirmed anomaly and enough legitimate route changes. We set the sliding window length to two hours for the same reason. 
The thresholds $th_d$ and $th_v$ are decided by historical data. Specifically, we use the observed path difference scores for all legitimate route changes two hours before the current window
as a reference distribution, and set $th_d$ as the knee point of its CDF curve.
$th_v$ is determined similarly. \update{The calculation of the knee point is automated via kneed~\cite{kneed}.}

To systematically evaluate our system, we substitute the \model engine with other commonly used features or representation learning models to create variants of our detection system.
In particular, we create 6 variants: 
ED uses the edit distance~\cite{ristad1998learning} to measure path difference; JI uses the Jaccard index of neighbor AS sets to measure the similarity of two ASes; 
Li, Ma, NV and SD uses the general network representation model Line~\cite{tang2015line}, Marine~\cite{feng2019marine}, node2vec~\cite{grover2016node2vec} and SDNE~\cite{wang2016structural} to train embedding vectors, respectively, and utilize Euclidean distance of embedding vectors to measure the similarity of two ASes. All these variants are well-trained and use the same settings as our detection system. Moreover, we compare our system with two state-of-the-art ML-based BGP anomaly detection approaches: AV~\cite{shapira2022ap2vec} and LS~\cite{dong2021isp}. AV and LS identify anomalous route changes and the time intervals where anomalies occur, respectively. For fair comparisons, we apply our anomaly detector to aggregate their detection results into different alarms (see Appendix~\ref{sec:appendix:comparison}).
\update{We do not compare our work with the active probing-based systems like~\cite{vervier2015mind,shi2012detecting} because they heavily rely on the real-time probing results collected by many data-plane facilities. These real-time probing results are not available for our historical datasets.}


Each detection system may raise multiple alarms for a dataset. Each alarm reports a potential routing anomaly. If any alarm matches the confirmed information, \ie the target prefix is reported as one of the anomalous prefixes and the misbehaved ASes are also identified as responsible, we consider the confirmed anomaly in this dataset as \emph{detected}.
Besides the confirmed anomaly, there could be other alarms that indicate unrevealed routing anomalies or are simply false alarms.
It is difficult to contact the operators to further confirm these potential anomalies since most anomalies occurred long 
ago.
\update{To verify these unknown alarms, we define four anomalous route change patterns that represent typical routing anomalies based on domain knowledge and authorized data such as RPKI validation states. If an alarm matches at least one pattern, we consider it a true alarm with high confidence, otherwise it is a \emph{false alarm}. These patterns are as follows:
\begin{itemize}[itemsep=0em,align=parleft,left=0pt..1em]
    \item P1 (Unauthorized Route Change): The origin ASes before and after the route change belong to different organizations and have different RPKI validation states, \ie one in the \emph{invalid\_ASN} state and the other in the \emph{valid} state~\cite{mohapatra2013bgp}.
    \item P2 (Route Leak): The routing path before or after the route change violates the valley-free criterion~\cite{gao2001inferring}.
    \item P3 (Path Manipulation): The routing path before or after the route change contains reserved ASNs or adjacent ASes that have no business relationship records between them~\cite{sermpezis2018artemis}.
    \item P4 (ROA Misconfiguration): The origin ASes before and after the 
    change are from the same organization but have different RPKI validation states, \ie one in the \emph{invalid\_length} or \emph{invalid\_ASN} state and the other in the \emph{valid} state~\cite{gilad2016we}.
\end{itemize}

\noindent
These patterns are endorsed by the domain experts from a large ISP where we deployed our system. They also apply these patterns to verify our real-world detection results (see \S\ref{sec:measurement:real-world-deployment}). 
Note that these patterns alone \emph{should not} be used to detect routing anomalies directly because they cannot correlate massive route changes with the same root cause, which would result in too many false alarms. We discuss the rationale behind these patterns in Appendix~\ref{sec:appendix:classify-detected-routing-anomaly-events}.
}


Table~\ref{tab:detection-result} shows that all previously-confirmed 18 routing anomalies are correctly detected by our system within tens of alarms. Further, our system reports no false alarms for 6 \update{datasets that cover the confirmed anomalies}, and only 5 false alarms in the worst case.
\update{These false alarms are mainly related to route engineering practices such as AS prepending~\cite{chang2005inbound}, while some involve stub ASes with limited connections.}
In contrast, the baselines 
cannot detect all these 
anomalies and raise more false alarms than ours (except for SD that cannot detect any anomalies).  
Besides, the baselines require many extra data to train the models, \eg AV needs RIB entries in every two hours and LS requires a large amount of training data to eliminate the negative impacts of label noises. Moreover, AV cannot detect transient anomalies and LS incurs high FP due to per-minute anomaly detection.
Please see the details in Appendix~\ref{sec:appendix:comparison}.
Overall, our detection system outperforms these baselines 
by significant margins. 
In summary, \first our system addresses the key challenges of ML-based detection methods and realize effective 
Internet routing anomaly detection. \second compared with our \model model, general network representation learning models are not able to effectively capture BGP semantics for routing anomaly detection. 

\subsection{Runtime Overhead}
\update{Our system runs on a Linux server with Intel Xeon E5-2650v4 (2.20GHz).}
We present its runtime overhead in Fig.~\ref{fig:overhead}. The X-axis displays the datasets in their chronological order\update{, with the later datasets containing more ASes in operation. The number of ASes in our datasets increases from 27,588 to 73,014 over the entire period.} The top figure shows the average time it takes to process every 15 minutes of data from each dataset. The error bar reflects the 95\% confidence interval. The variance in processing time is mostly caused by the variance in the sizes of 15-minute data. \update{The processing time is less than 100 seconds in most cases, indicating that the increase of ASes only slightly impacts our system's runtime overhead.}
The largest processing time ($\sim$140s in $RL_{jtl}$) is still much smaller than 15 minutes, meaning that our system can effectively process the stream of global route announcements in real time. The bottom figure plots the average time it takes to process every 1,000 route changes. Due to the caching employed by our system (\eg the caching of path difference scores), the anomaly detector spends much less time at the steady state ($\sim$0.05s per 1,000 route changes), compared with the cold start ($\sim$0.2s).

\begin{figure}[t]
  \centering
    \includegraphics[width=0.9\linewidth]{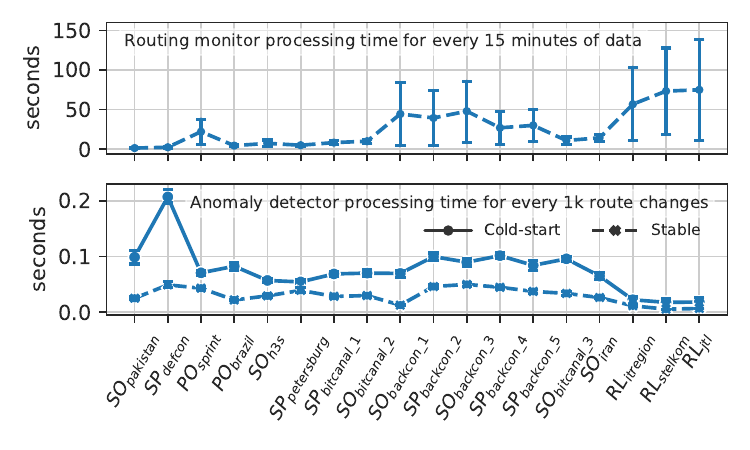}
  \caption{\textbf{The runtime overhead of our detection system.}}
  \label{fig:overhead}
\end{figure}

\subsection{Robustness Analysis}

\update{We analyze the robustness of our system given noisy AS relationship data, which is created by modifying or deleting the original AS relationships in the CAIDA dataset. We consider four types of noisy datasets. In particular, given a noise ratio $r$, we first randomly select $r\%$ AS relationships, and then flip their labels (\ie changing P2P to P2C and P2C to P2P) to produce the R1 dataset, or delete them 
to create the R2 dataset. R1 and R2 represent the noisy dataset caused by inaccurate and incomplete AS relationship inference, respectively. To create another two types of noisy datasets, we first select top-$r\%$ AS relationships with the fewest BGP routes that use their underlying AS-to-AS links and create the W1 dataset by flipping their labels and W2 by deleting them. These two types of noisy data are common in the Internet because the AS relationships serving fewer routes are more likely to be mislabeled due to their limited Internet visibility.

We train our system with each noisy dataset independently and evaluate its detection performance following the same steps in \S\ref{sec:detection:results}. We repeat each experiment for 5 times to avoid bias and plot the results in Fig.~\ref{fig:robustness}. The error bar shows the 95\% confidence interval. Even when the noise ratio reaches 20\%, our system still detects at least 17 true anomalies (18 in total) and only generates less than 40 false alarms across all datasets. Note that such a high noise ratio is rare in practice. These results demonstrate the robustness of our system under noisy AS relationship data.
}

\begin{figure}[t]
  \centering
    \includegraphics[width=0.9\linewidth]{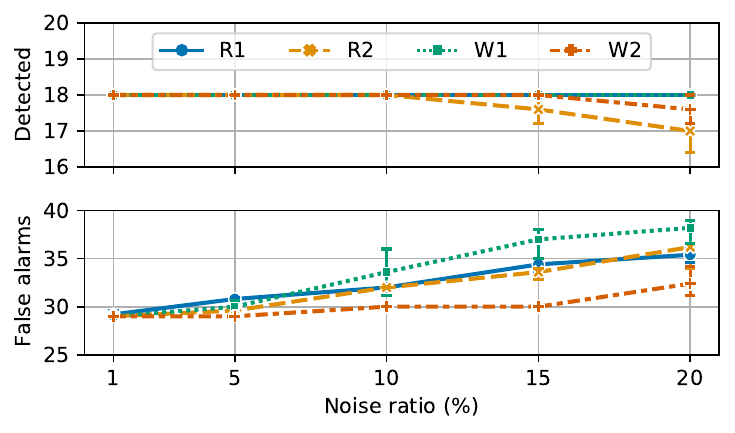}
  \caption{\textbf{\update{Detection results
  given noisy AS relationships.
  }}}
  \label{fig:robustness}
\end{figure}

\subsection{Real-World Deployment}
\label{sec:measurement:real-world-deployment}
We deploy our detection system in the main operational AS of a large ISP to evaluate its performance in practice. Based on the customer cone size, the AS rank is in the global top 100~\cite{caida_as_rank}. At the time of our deployment, the AS maintains live BGP sessions with about 500 neighbors, including 14 Tier-1 ASes. Thereby, the AS has a fairly comprehensive view on the Internet-wide routing paths. 
We set up our detection system in a server that receives real-time incoming route announcements from all BGP routers within the AS via the iBGP protocol. We train the \model engine of our system using the latest CAIDA AS relationship dataset collected at the time of our deployment, which includes 74,923 ASes and 505,927 AS business relationship records. The sliding window length of our system is one hour. We use the same parameter settings and the method of identifying false alarms as described in \S\ref{sec:detection:results}. To reduce the repetitive alarms raised in different time windows, we aggregate the alarms sharing the same anomalous prefix and responsible AS.


The system is online since Jaunary 1, 2023. 
We analyze the generated results from January 1 to February 1, 2023. In total, the system processes 152,493,303 live route announcements during this month, detects 5,106,442 route changes and raises 548 alarms. We show the alarms' impact and daily statistics in Table~\ref{tab:anomaly-impact} and Fig.~\ref{fig:daily-number}, respectively. On average, our system identifies 17.68 alarms per day.
\update{The domain experts of the ISP carefully verify the correctness of each alarm based on the patterns described in \S\ref{sec:detection:results}}.
\update{They find that most alarms (497 out of 548) indicate real routing anomalies, \ie true alarms. These true alarms, not detected by the ISP’s existing routing security mechanisms,  include 84 unauthorized route changes (P1), 123 route leaks (P2), 270 path manipulations (P3) and 20 ROA misconfiguration (P4).
}
The interpretability of these alarms greatly facilitates the identification of anomaly sources (see \S\ref{sec:measurement:case-study} for detailed case study). More importantly, our system only raises an average of 1.65 false alarms per day.
These false alarms can be further eliminated with minimal intervention (see discussions in \S\ref{sec:discussion}).
Overall, our detection system demonstrates promising results in real-world deployment. 

\update{
We further investigate \emph{how many anomalies are caught by an existing security mechanism but not our system}. First, we check whether our system misses the anomalies detected by the ISP's existing security mechanism. The ISP detects invalid routes based on its customers' IRRs. However, due to the incomplete coverage of prefixes, the ISP's security mechanism does not generate any alarm during our system's deployment. Thus, our system does not miss any alarm raised by the ISP itself. Next, we check whether there are real routing anomalies missed by both our system and the ISP itself. This requires a reliable source for BGP incidents. Although BGPstream is a promising candidate, we cannot use it in our paper because its vantage points are quite different from those of the ISP, which would introduce non-negligible experimental bias. Therefore, we use the RPKI validation results as a reference. Specifically, among all the RPKI-invalid announcements that are generated during our system's deployment, our system only misses about 2.25\% of them and most of the missed ones are due to the limited Internet visibility, \ie only a few vantage points observe these invalid announcements. Overall, this result indicates our system has very low false negatives.}

\noindent \textbf{Ethics.} We operate our detection system in compliance with the ISP's policy/agreement and under the close supervision of ISP's administrators. Our evaluation does not involve any sensitive or privacy data. We only collect results for analysis and do not interfere with Internet routing operations. 

\begin{table}[t]
    \caption{\textbf{The overall impact of the detected anomalies.}}
    \label{tab:anomaly-impact}
    \small
    \centering
    \begin{tabu}{X[1,c] X[1,c] X[1,c]}
    \toprule
    \textbf{\#Affected Routes}&\textbf{\#Affected Prefixes}&\textbf{\#Affected Origins} \\
    \midrule
    1,202&961&477 \\
    \bottomrule
    \end{tabu}
\end{table}

\begin{figure}[t]
  \centering
    \includegraphics[width=0.9\linewidth]{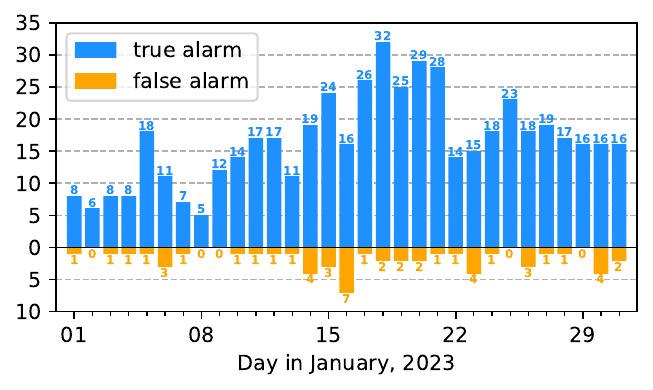}
  \caption{\textbf{Daily alarm number in real-world deployment.}}
  \label{fig:daily-number}
\end{figure}

%% file: content/sec6-case-study.tex
\section{Case Study}
\label{sec:measurement:case-study}

\begin{figure*}[!ht]
  \centering
    \includegraphics[width=0.8\linewidth]{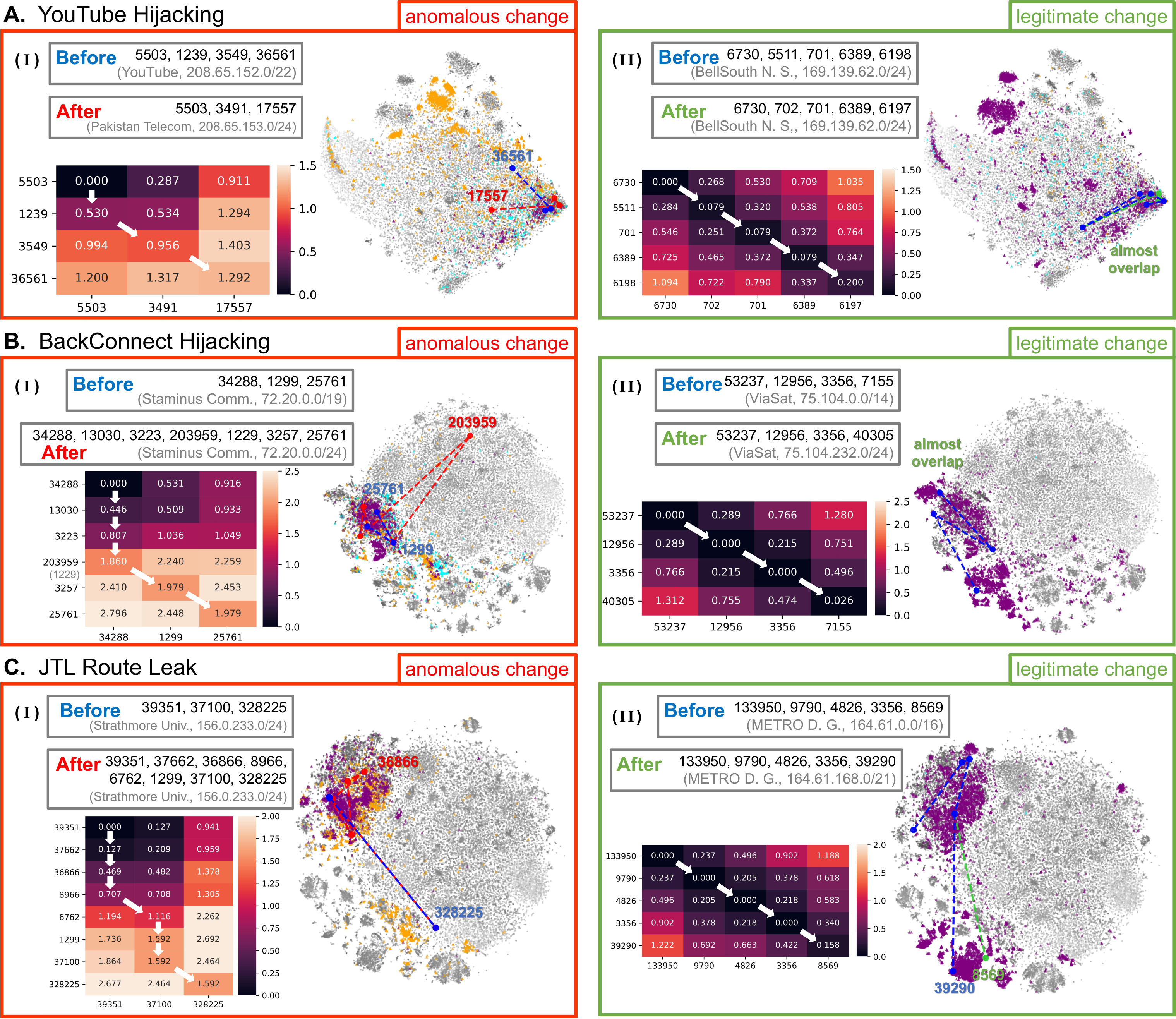}
  \caption{\textbf{
  Typical anomalous/legitimate route changes.
  }
  \textnormal{In each plot, the heat map (left)
  displays the path difference score at each AS pairing step of DTW in a top-to-bottom and left-to-right order.
  The white arrows show the optimal sequence of AS pairs. The embedding map (right) represents embedding vectors in a 2-D plane.
  The purple dots denote the common neighbors of the two routing paths, while the orange and cyan dots correspond to the exclusive neighbors for the new and old paths, respectively.
  }
  }
  \label{fig:heatmap}
\end{figure*}

In this section, we illustrate the interpretability of our detection results by analyzing four detected anomalies: three from historical events (Cases 1-3) and one from our real-world deployment (Case 4). Cases 1-3 each cover a different category of routing anomalies, \ie origin change, path change, and route leak. In each case, we compare a representative anomalous route change with a legitimate one that occurred on the same day.
To interpret the difference between routing paths, 
we apply two visualization techniques: the heat map of the path difference scores and the embedding map of the routing roles.
In particular, the heat map shows 
the path difference scores for the eligible sequences of AS pairs and marks the optimal sequence identified by DTW. 
The embedding map, generated by t-SNE~\cite{van2008visualizing}, 
visualizes the routing roles of the ASes 
to illustrate the deviations between 
two paths.

\noindent\textbf{Case 1.}
\mbox{Figure~\ref{fig:heatmap}(A)} shows the origin change in $SO_{pakistan}$, where
AS 17557 (\emph{Pakistan Telecom}) hijacked a subprefix of AS 36561 (\emph{YouTube}) by announcing 208.65.153.0/24. 
\mbox{Figure~\ref{fig:heatmap}(A)(I)} sees the anomalous pattern
that most ASes on the new path are different from those on the old path.
The two 
origin ASes, AS 17557 and AS 36561, are
far apart
in the embedding map 
with few common neighbors,
indicating very different 
routing roles. 
Thus, their route announcements would traverse quite different paths before   
they eventually converge on AS 5503. This pattern also appears 
in the heat map, where the path difference score 
rises sharply as the AS pairing operation proceeds, because the ASes from two paths are completely different after a few steps.
In contrast, \mbox{Fig.~\ref{fig:heatmap}(A)(II)} shows a legitimate route change on the same day with low path difference scores and almost overlapping paths in the embedding map,
because the legitimate route update changes the origin from AS 6198 to AS 6197,
both operated by \emph{BellSouth Network} and with similar routing roles.

\noindent\textbf{Case 2.}
\mbox{Figure~\ref{fig:heatmap}(B)} shows the path change in $SP_{backcon\_2}$, where
AS 203959 (\emph{BackConnect}) hijacked a subprefix of AS 25761 (\emph{Staminus Comm.}) by faking a nonexistent routing path to the real origin AS. 
The embedding map reveals that the fake path detours significantly from the real one, as AS 203959 is neither on the real path nor similar to any ASes on it in terms of routing roles.
The heat map also indicates an upsurge in the path difference score 
in the AS pairing.
The optimal sequence of AS pairs greatly deviates from the diagonal of the heat map because no AS is similar to AS 203959 in the counterpart.
In contrast, the paths in the legitimate route change are similar in routing roles; we leave the analysis to the reader.

\noindent\textbf{Case 3.}
\mbox{Figure~\ref{fig:heatmap}(C)} shows the route leak in $RL_{JTL}$, where
AS 36866 (\emph{JTL}) received the route to 156.0.233.0/24 from its provider AS 8966 (\emph{Emirates Tel.}) and leaked it to its another provider AS 37662 (\emph{WIOCC}). 
The leaked route results in a much longer path through AS 37662 that detours significantly from the original one. It also violates the valley-free criterion. Accordingly, the heat map shows high path difference scores and the optimal sequence forms a hump-like pattern away from the diagonal.
In contrast, the legitimate route change has minor path difference; we leave the analysis to the reader.

\noindent\textbf{Case 4.}
We represent the first alarm reported by our system during its real-world deployment in Fig.~\ref{fig:anomaly-report}. 
This alarm, labeled Alarm 0, starts at 01:04:40, January 1, 2023, and lasts about 1 hour and 24 minutes, affecting three prefixes and six routes observed by two individual vantage points. 
The top part of Fig.~\ref{fig:anomaly-report} provides an overview of this alarm. 
The bottom part of Fig.~\ref{fig:anomaly-report} further plots
a representative route change event captured by this alarm observed from AS 6453. In this event, AS 42440 (\emph{RDG-AS}) announces 185.88.179.0/24, which is owned by AS 201691 (\emph{WEIDE}), without authorization (its RPKI validation state is \emph{invalid\_ASN}). Moreover, AS 42440 is also on the path before the route change, indicating a Type-5 route leak as described in RFC 7908~\cite{sriram2016problem}. This pattern also appears in the heat map. The optimal sequence before AS 42440 follows the diagonal and the path difference score increases sharply after AS 42440, indicating a significant change of routing roles. The embedding map also demonstrates the clear difference between AS 42440 and AS 201691. 


\begin{figure}[t]
  \centering
    \includegraphics[width=0.9\linewidth]{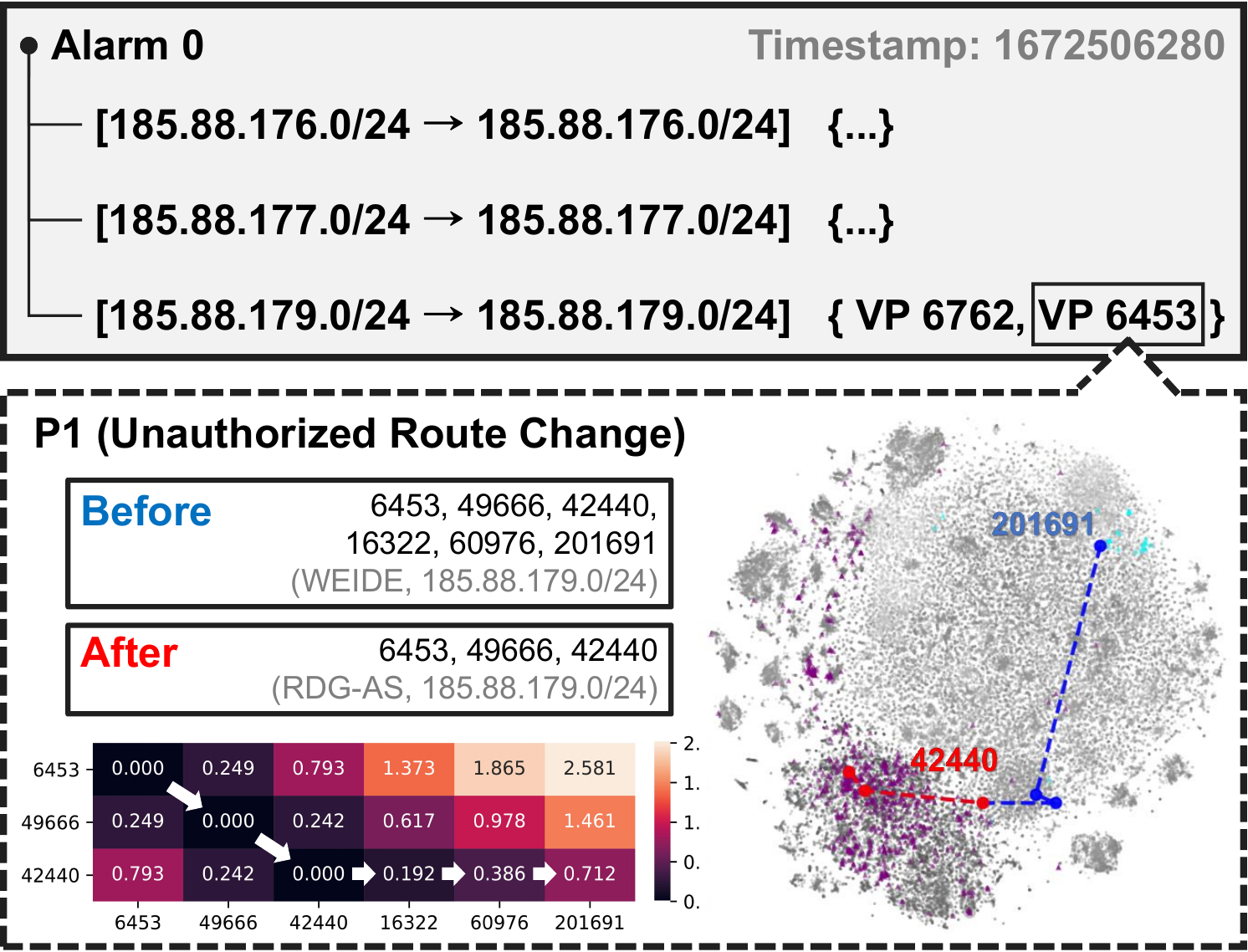}
  \caption{\textbf{An anomaly report in real-world deployment.}}
  \label{fig:anomaly-report}
\end{figure}

%% file: content/sec7-discussion.tex
\section{Discussion}
\label{sec:discussion}

\update{
\noindent \textbf{Per-Prefix Threshold.} 
Our system uses the same detection threshold values (\ie $th_d$ and $th_v$) for all prefixes. 
Assigning an individual threshold for each prefix may gain better results. But it will incur significant computational overheads and large detection delays due to the large number of prefixes in the Internet (\ie over 940k IPv4 and 200k IPv6 prefixes in November 2023), which is inappropriate for online detection.


}

\update{
\noindent \textbf{Reducing False Alarms.} 
We can apply two heuristics to further reduce our system's false alarms. \first Since most route hijacks and route leaks are short-lived~\cite{vervier2015mind}, we can label the long-lived routes as normal. \second The different origin ASes in a legitimate MOAS event usually have similar routing characteristics~\cite{zhao2001analysis} (\eg from the same organization) while the malicious ASes do not.
We leave the application of both heuristics in our system to future work.

}

\update{
\noindent \textbf{Adaptive Attacks.} 
An adaptive attacker may attempt to bypass our system by mimicking the normal ASes' routing roles, yet it is very difficult in practice. Specifically, to imitate a normal AS's routing role, a malicious AS has to establish business relationships with the normal AS's neighbors, which is time-consuming and would reduce the malicious AS's stealthiness due to the possible scrutiny required for establishing neighboring relationships.  
Further, to make the routing roles of a forged path similar to those of a real path, the attacker has to control multiple ASes on the forged path. Our empirical study (presented in detail in Appendix~\ref{sec:appendix:adaptive-attacks}) reveals that it requires the attacker to control at least two ASes and establish hundreds of new AS relationships to make the routing roles of a forged path similar to those of a real path. This is considerably intractable in practical scenarios. 
}

\noindent \textbf{Evolving Routing Roles.} AS routing roles evolve
as ASes update routing policies.
Our system is resilient against routing role evolution.
Per \S\ref{sec:measurement:analysis-results}, even if the time gap between training and detection is over 30 days (\ie the dataset update cycle), our system still performs well.
Moreover, we can retrain \model with latest AS relationships
to keep up with the routing role evolution. It takes 
\(\sim\)10 hours
to train the model on our platform with
GeForce RTX 2080 Ti, which is acceptable 
since CAIDA releases a new dataset roughly
every month. 

\noindent\textbf{Detection with Unknown ASes.} \model cannot learn routing roles of ASes whose relationships with other ASes are unknown (\ie unknown ASes). Fortunately, the existing study \cite{luckie2013relationships} has revealed AS relationships for most ASes absent in the 
dataset. 
For instance, there are only 368 unknown ASes in the most recent events that we analyze, which is only 0.5040\% of the total ASes. 
Further, our measurement study discovered that some unknown ASes use the AS numbers that are reserved for private use
~\cite{mitchell2013autonomous}. The routing paths containing such unknown ASes should be regarded as anomalous. 


%% file: content/sec8-related-work.tex
\section{Related Work}
\label{sec:related-work}

\noindent \textbf{Traditional Routing Anomaly Detection.}
The existing studies 
consist of
the control-plane based~\cite{sermpezis2018artemis,yan2009bgpmon,li2005internet}, the data-plane based~\cite{zhang2008ispy,li2012buddyguard,zheng2007light} and the hybrid methods~\cite{schlamp2016heap,shi2012detecting,vervier2015mind,hu2007accurate}.
The control-plane based methods maintain the normal/authoritative route information for each prefix and check if the newly received route contradicts it. The data-plane based methods identify routing anomalies by analyzing the reachability from multiple hosts to the 
target prefix.
However, these traditional approaches require non-trivial manual investigation from network operators, \eg collecting the normal route information of each prefix, and deploying vast
network probes to monitor the prefixes in the world, which incurs unacceptable operation overhead for deployment.

\noindent \textbf{ML Based Routing Anomaly Detection.} 
Machine learning is utilized to detect routing anomalies ~\cite{cheng2016ms,testart2019profiling,cheng2018multi,dong2021isp,al2015detecting,al2012machine,lutu2014separating,deshpande2009online,theodoridis2013novel,shapira2020deep,hoarau2021suitability,shapira2022ap2vec,sanchez2019comparing}.
Shapira \etal~\cite{shapira2022ap2vec} (AV) use unsupervised word embedding to model routes.
\update{But AV fails to detect transient anomalies due to its reliance on RIB snapshots and suffers from non-trivial retraining overheads.}
Dong \etal~\cite{dong2021isp} and Hoarau \etal~\cite{hoarau2021suitability} perform supervised classification on BGP time series. Both require large-scale
labeled datasets, which is hard to achieve in practice. 
Moreover, all these methods cannot provide interpretable
results,
incurring great manual efforts for validation. 

\noindent \textbf{Application of Network Representation Learning.}
Network representation learning (NRL)~\cite{zhang2018network} aims to learn latent, low-dimensional representations of network vertices, while preserving network characteristics, \eg the structure information and vertex content, where the learnt representations can be used for downstream tasks. NRL has been applied in various domains, \eg social network analysis~\cite{qiu2019noise,liu2019structural,wang2019online}, recommendation system~\cite{tan2020learning,ge2020graph,liu2020multi}, and anomaly detection~\cite{ye2019out,zhang2020gcn,fan2021heterogeneous}. 
In this paper, we develop a BGP semantics aware NRL model to measure ASes' routing roles for routing anomaly detection.

%% file: content/sec9-conclusion.tex
\section{Conclusion}

In this paper, we present a routing anomaly detection system centering around a novel network representation learning model named \model. The core design of \model is to accurately learn the routing roles of Internet ASes by incorporating the BGP semantics. As a result, routing anomaly detection, given \model, is reduced to discovering unexpected routing role churns upon observing new route announcements. We implement a prototype of our routing anomaly detection system and extensively evaluate its performance using 18 real-world RouteViews datasets containing over 11 billion route announcement records. The results demonstrate that our system can detect all previously-confirmed routing anomalies 
within an acceptable number of alarms. 
We also perform one month of real-world detection at a large ISP and detect 497 true anomalies in the wild with only 1.65 daily false alarms on average, demonstrating the practical feasibility of our system.

%% file: content/supplementary.tex
\section{Preservation of the Second-Order Proximity}
\label{sec:appendix:proof-of-the-second-order-proximity-preservation}

In practice, we derive the weight vector $\bm{l}$ (see Def.~\ref{def:first-order-pro}) from the softmax transformation of a learnable variable $\bm{l'} \in \mathbb{R}^d$, and ensures that the minimum value of its components is above $\alpha$ $(\alpha > 0)$ to avoid trivial solutions:

\begin{equation}
    \bm{l} = [l_0, \dots , l_{d-1}]^{\intercal} = \big(softmax(\bm{l'})\big)_{\alpha+}. 
\end{equation}

\noindent
Accordingly, $\bm{l}$ is differentiable with respect to $\bm{l'}$ and satisfies the following criteria:
\begin{equation}
    \|\bm{l}\|_1=\sum_{i=0}^{d-1}|l_i|=1,\quad
    l_{min}=\min\{l_0, \dots , l_{d-1}\} \ge \alpha \label{eq:weight}.
\end{equation}

\noindent
In practice, we set $\alpha = \frac{1}{d} \times 10^{-6}$, a rather small positive number.
Now, we prove the theorem that the pairwise AS difference does preserve the characteristics of the second-order proximity between ASes\footnote{Note that this is a theoretical proof under ideal conditions. In the actual training process, the convergence of the embedding vectors would also be affected by initial values and sampling techniques.}. 

\begin{theorem}
    The distance function~\eqref{eq:distance} preserves the second-order proximity. 
    \label{thm:second-order-similarity}
\end{theorem}
\begin{proof}
    Given two vertices $u, v$ $(u \neq v)$ having similar neighbors and the corresponding business relationships, there must exist at least one vertex $w$ $(w \neq u, w \neq v)$ such that $(u, w) \in E, (v,w) \in E$.
    Then after sufficient optimization of the objective function, there must exist two thresholds $\epsilon_1,\epsilon_2>0$ such that
    $$0 \le p\_score(u, w), p\_score(v,w) \le \epsilon_1,$$ and 
    $$h\_score(u,w), h\_score(v, w) \ge \epsilon_2.$$
    According to Equation~\eqref{eq:weight}, we have
    \begin{align}
        p\_score(u, w) &\le \epsilon_1 \\
        &\Longrightarrow (\bm{x_w}-\bm{x_u})^{\intercal}((\bm{x_w}-\bm{x_u})\odot \bm{l}) \le \epsilon_1.  \\
        &\Longrightarrow \|\bm{x_w}-\bm{x_u}\|_2 \le \sqrt{\frac{\epsilon_1}{l_{min}}} \le \sqrt{\frac{\epsilon_1}{\alpha}} \label{eq:s1-i-v-norm}. \\
        p\_score(v, w) &\le \epsilon_1 \\
        &\Longrightarrow (\bm{x_w}-\bm{x_v})^{\intercal}((\bm{x_w}-\bm{x_v})\odot \bm{l}) \le \epsilon_1.  \\
        &\Longrightarrow \|\bm{x_w}-\bm{x_v}\|_2 \le \sqrt{\frac{\epsilon_1}{l_{min}}} \le \sqrt{\frac{\epsilon_1}{\alpha}}. \label{eq:s1-j-v-norm}
    \end{align}
    When it comes to $p\_score(u, v)$, the upper bound is determined according to Equation~\eqref{eq:s1}:
    $$p\_score(u, v) \le p\_score(u, w)+p\_score(v,w) \le 2\epsilon_1$$
    When it comes to $h\_score(u, v)$, the upper bound is determined according to Equation~\eqref{eq:s2}, \eqref{eq:s1-i-v-norm} and \eqref{eq:s1-j-v-norm}:
    \begin{align*}
        |h\_score(u, v)| &= |h\_score(u,w) - h\_score(v,w)| \\
                    &\le |h\_score(u,w)| + |h\_score(v,w)| \\
                    &= |(\bm{x_w}-\bm{x_u})^{\intercal}\bm{r}|
                        +|(\bm{x_w}-\bm{x_v})^{\intercal}\bm{r}| \\
                    &\le \|\bm{x_w}-\bm{x_u}\|_2 \cdot \|\bm{r}\|_2 
                        +\|\bm{x_w}-\bm{x_v}\|_2 \cdot \|\bm{r}\|_2 \\
                    &\le 2\sqrt{\frac{\epsilon_1}{\alpha}}. 
    \end{align*}
    Thus, we have $D_{\bm{l},\bm{r}}(u, v)= |p\_score(u, v)|+|h\_score(u, v)| \le 2\epsilon_1+2\sqrt{\frac{\epsilon_1}{\alpha}}$.
    That is, there exists an upper bound for the distance between $u, v$ in terms of the distance function~\eqref{eq:distance}. Moreover, when the objective function is sufficiently optimized, $\epsilon_1$ is a rather small positive number approaching zero and $\alpha$ is a constant, and then the distance $D_{\bm{l},\bm{r}}(u, v)$ should approach zero, which means the two vertices $u, v$ are at a close distance to each other in the low-dimensional vector space.
    Hence the proof completes.
\end{proof}

\section{Sampled AS Pair Datasets}
\label{sec:appendix:as-pair-datasets}

According to \S\ref{sec:method:embedding-results}, we sample AS pairs in three categories and obtain 15 datasets, as presented in Table~\ref{tab:sampling-sets}.

\begin{table}[ht]
\caption{\textbf{Sampled AS pair datasets.} \textnormal{\emph{No rel} is no relationship. \emph{Ngbr JI} refers to the Jaccard index of two neighbor AS sets.}}
\label{tab:sampling-sets}
\small
\centering
\begin{tabu}{X[4,c] X[2,c] X[4,c] X[12,l]}
    \toprule
    \textbf{Feature}&\textbf{Name}&\textbf{\#AS pair}&\textbf{Sampling rule} \\
    \midrule
    \multirow{3}{*}{\shortstack[c]{1st-order \\ proximity}}
        &$S_0$&10,000&P2P, $\textrm{Ngbr JI} \in [0\%,10\%)$\\
        &$H_0$&10,000&P2C, $\textrm{Ngbr JI} \in [0\%,10\%)$\\
        &$N_0$&10,000&no rel, $\textrm{Ngbr JI} \in [0\%,10\%)$\\
    \midrule
    \multirow{5}{*}{\shortstack[c]{2nd-order \\ proximity}}
        &$N_1$&2,000&no rel, $\textrm{Ngbr JI} \in [0\%,20\%)$\\
        &$N_2$&2,000&no rel, $\textrm{Ngbr JI} \in [20\%,40\%)$\\
        &$N_3$&2,000&no rel, $\textrm{Ngbr JI} \in [40\%,60\%)$\\
        &$N_4$&2,000&no rel, $\textrm{Ngbr JI} \in [60\%,80\%)$\\
        &$N_5$&2,000&no rel, $\textrm{Ngbr JI} \ge 80\%$\\
    \midrule
    \multirow{7}{*}{hierarchy}
        &$S_1$&10,000&with a direct P2P\\
        &$H_1$&10,000&with a direct P2C\\
        &$H_2$&10,000&with 2 consecutive P2C\\
        &$H_3$&10,000&with 3 consecutive P2C\\
        &$H_4$&10,000&with 4 consecutive P2C\\
        &$H_5$&10,000&with 5 consecutive P2C\\
        &$H_6$&10,000&with 6 consecutive P2C\\
    \bottomrule
\end{tabu}
\end{table}

\section{Real-World Route Announcement Datasets}
\label{sec:appendix:route-datasets}

We collect 18 real-world route announcement datasets from RouteViews Project~\cite{routeviews}. Each dataset contains all route announements from all RouteViews collectors 12 hours before and after a previously-confirmed routing anomaly. We name each dataset (\eg $PO_{brazil}$) by the category and the abbreviation of the confirmed routing anomaly related with the dataset. Table~\ref{tab:bgp-event-dataset} shows the details of each dataset.

\begin{table*}[!ht]
    \caption{\textbf{Real-world datasets with previously-confirmed routing anomalies.} \textnormal{$PO$, $SO$, $SP$ and $RL$ indicate the categories of anomalies, \ie the prefix origin change, the subprefix origin change, the subprefix path change, and the route leak, respectively.
    \emph{\#Total Ann} denotes the number of announcements in each dataset. \emph{GT Duration}, \emph{\#GT Ann}, and \emph{\#GT VP} denote the duration of the ground-truth anomaly, the number of anomalous announcements in the ground-truth anomaly, and the number of vantage points observing the ground-truth anomaly, respectively.}
    }
    \label{tab:bgp-event-dataset}
    \small
    \centering
    \begin{tabu}{X[4,c] X[3,l] X[5,c] X[3,r] || X[3,r] X[2,r] X[2,r]}
        \toprule
        \textbf{Category}&\textbf{Name}&\textbf{Time Span ($\bm{\pm12h}$)}&\textbf{\#Total Ann}&\textbf{GT Duration}&\textbf{\#GT Ann}&\textbf{\#GT VP} \\
        \midrule
        \multirow{2}{*}{\shortstack[c]{prefix hijacking \\ (origin change)}}
&$PO_{brazil}$&2014-09-10 00:30:00&93,942,111&1h28m39s&127&24\\
&$PO_{sprint}$&2014-09-09 13:45:00&91,338,780&7s&198&74\\
        \midrule
        \multirow{7}{*}{\shortstack[c]{subprefix hijacking \\ (origin change)}}
&$SO_{iran}$&2018-07-30 06:15:00&201,516,610&3h25m42s&587&102\\
&$SO_{bitcanal\_3}$&2018-06-29 13:00:00&141,856,933&47m34s&672&107\\
&$SO_{backcon\_3}$&2016-02-21 10:00:00&279,389,745&4s&1,156&100\\
&$SO_{backcon\_1}$&2015-12-03 22:00:00&186,952,580&16m5s&695&96\\
&$SO_{bitcanal\_2}$&2015-01-23 12:00:00&60,327,936&5m11s&284&34\\
&$SO_{h3s}$&2014-11-14 23:00:00&94,519,947&2s&581&86\\
&$SO_{pakistan}$&2008-02-24 18:00:00&8,155,604&4h57m6s&156&20\\
        \midrule
        \multirow{6}{*}{\shortstack[c]{subprefix hijacking \\ (path change)}}
&$SP_{backcon\_5}$&2016-05-20 21:30:00&147,824,105&1h33m47s&296&61\\
&$SP_{backcon\_4}$&2016-04-16 07:00:00&100,632,725&11m28s&1,032&98\\
&$SP_{backcon\_2}$&2016-02-20 08:30:00&314,047,522&3m10s&825&73\\
&$SP_{bitcanal\_1}$&2015-01-07 12:00:00&57,902,559&12m55s&286&38\\
&$SP_{petersburg}$&2015-01-07 09:00:00&60,016,016&28s&1,000&83\\
&$SP_{defcon}$&2008-08-10 19:30:00&5,701,271&26s&72&21\\
        \midrule
        \multirow{3}{*}{\shortstack[c]{route leak}}
&$RL_{jtl}$&2021-11-21 06:30:00&3,035,292,256&1h14m47s&2,419&146\\
&$RL_{stelkom}$&2021-11-17 23:30:00&3,572,332,877&3m18s&1,888&140\\
&$RL_{itregion}$&2021-11-16 11:30:00&3,409,628,374&31m32s&2,409&144\\
        \bottomrule
    \end{tabu}
\end{table*}

\section{Identifying Legitimate Route Changes}
\label{sec:appendix:legitimate-route-change}

We identify legitimate route changes via a heuristic rule derived from the multi-homing settings~\cite{zhao2001analysis}, which is a common route engineering practice.
We label an origin change as legitimate if it satisfies two conditions: (i) Two different origin ASes (\ie on the changed and the original routing paths, respectively) belong to the same organization; 
(ii) There is no duplicate AS numbers in its routing paths. 
The first condition ensures the legitimacy of this route change since the actual ownership of the corresponding prefix is not changed, and the second one filters out the potential noises caused by routing paths manipulations, \eg AS prepending. 
Note that we use the CAIDA AS organization dataset~\cite{caida_as_organizations}
to obtain the organization name of each AS.
These rules can not identify all legitimate route changes in the Internet. We only utilize them to find out changes with high confidence.


\section{Detection System built upon BEAM}
\label{sec:appendix:detection-system}

\begin{algorithm}[t]
    \caption{Measuring Path Difference Scores}
    \begin{algorithmic}[1]
        \Function{PathDiffScore}{$S$: Array $[1 \dots m]$, $S'$: Array $[1 \dots n]$}
            \State \textbf{external function} $D_{\mathbf{l},\mathbf{r}}$
            \State DIFF $\coloneqq $ Array $[0 \dots m, 0 \dots n]$
            \For{$i \coloneqq 0$ to $m$ }
                \For{$j \coloneqq 0$ to $n$}
                    \State DIFF$[i,j] \coloneqq +\infty$
                \EndFor
            \EndFor
            \State DIFF$[0,0] \coloneqq 0$
            \For{$i \coloneqq 1$ to $m$}
                \For{$j \coloneqq 1$ to $n$}
                    \State diff $\coloneqq D_{\mathbf{l},\mathbf{r}}(S[i], S'[j])$
                    \State{$\begin{aligned}
                        \text{DIFF}[i,j] \coloneqq \text{diff} + \text{minimum}(\text{DIFF}[i-1,j], \\
                            \text{DIFF}[i, j-1], \text{DIFF}[i-1, j-1])
                        \end{aligned}$
                        }
                \EndFor
            \EndFor
            \State \Return $\text{DIFF}[m, n]$
        \EndFunction
    \end{algorithmic}
    \label{alg:as-path-difference}
\end{algorithm}


Our detection system built upon \model consists of three components, \ie the routing monitor, the \model engine, and the anomaly detector. At a high level, the routing monitor detects route changes from the BGP route update announcements. The \model engine utilizes a pre-trained \model model to compute the path difference scores of the route changes. The detailed path difference score computation process is presented in \S\ref{sec:measurement:path-differnce-score} and we supplement its pseudo-code in Algorithm \ref{alg:as-path-difference}. Besides, the anomaly detector is responsible for identifying routing anomalies from the route changes and generating corresponding alarms. Here, we will introduce the details of the routing monitor and anomaly detector.   


\noindent\textbf{The Routing Monitor.} This component collects the BGP route update announcements from global vantage points, to detect route changes. In particular, it maintains a routing table for each vantage point. To reduce the storage space and facilitate detection, the routing table is in a trie structure, where each node represents a unique prefix and its routing path. Also, the prefix of each parent node is the super prefix of its child nodes. When the routing monitor receives a new route update message that announces routing path $l$ to prefix $p$ from a vantage point, it searches $p$ in the corresponding trie structure routing table. If a routing path to prefix $p$ (denoted by $l'$) already exists in the routing table, we compare it with $l$ and find a route change when $l \neq l'$. If $p$ is not in the routing table, we compare $l$ with the routing path to the most specific super prefix of $p$, which can be obtained from the parent node of $p$. After detecting route changes, the routing monitor updates the routing table with the newly received update message.

\noindent\textbf{The Anomaly Detector.} This component aims to detect routing anomalies from the route changes and generate corresponding alarms. Its detection procedure includes three steps: detecting suspicious route changes, identifying anomalous prefixes, and locating responsible ASes.

Formally, we represent each detected route change by $(t, r, p, p', l, l')$, where $t$ is the occurrence time of the route change, $r$ is the AS number of the vantage point that captures the route change, $p$ is the prefix announced in the update announcement that triggers the route change, $p'$ is the prefix in the routing table that conflicts with $p$ (\ie the same as $p$ or the most specific super prefix of $p$), and $l$ and $l'$ are the routing path to $p$ and $p'$, respectively. The computed path difference score between $l$ and $l'$ is denoted as $d_{l,l'}$. If $d_{l,l'}$ is above a threshold $th_d$ computed based on historical legitimate route changes (detailed in \S\ref{sec:detection:results}), the anomaly detector flags the route change as \emph{suspicious}.

It is necessary to prioritize the widespread routing anomalies captured by multiple vantage points, otherwise the transient route changes in the Internet will generate numerous insignificant alarms. Towards this end, the anomaly detector further groups suspicious route changes that share the same $(p,p')$ and denotes each group as a \emph{prefix event}, where each event is associated with a specific $(p,p')$ and sorts the suspicious route changes by their occurrence time. To each event, the anomaly detector utilizes a sliding window (with time interval $w$) to count the number of individual vantage points that observe the suspicious route changes within the window. If the counted number (denoted by $N_r$) is above another threshold $th_v$, the anomaly detector considers the prefix event is caused by a widespread routing anomaly and regards it as anomalous. The formal definition of anomalous prefix event is:

\begin{definition}[Anomalous Prefix Event]
    Given parameters $th_d, w, th_v$, a prefix event of $(\hat{p},\hat{p}') \equiv P = \{(t, r, p, p', l, l') \mid p=\hat{p},p'=\hat{p}',d_{l,l'}>th_d\}$ is anomalous if and only if $\exists A \subseteq P, \forall (t_1, r_1, p_1, p'_1, l_1, l'_1),(t_2, r_2, p_2, p'_2, l_2, l'_2) \in A, |t_1-t_2|<w, N_r = |\{r \mid (t,r,p,p',l,l') \in A\}|>th_v$.
\label{def:anomalous-prefix-event}
\end{definition}

In realistic circumstances, a misbehaved ASes may impact multiple prefixes simultaneously, \eg AS 55410 hijacked more than 30,000 prefixes in April 2021, resulting in several different anomalous prefix events. To provide network administrators with comprehensive information about the impacted prefixes of each routing anomaly, the anomaly detector correlates all anomalous prefix events based on their responsible ASes. Intuitively, the ASes responsible for a suspicious route change (where $l=\{AS_1,AS_2,\dots,AS_n\},l'=\{AS'_1,AS'_2,\dots,AS'_m\}$) are either within the ASes present in $l$ but absent in $l'$ (\ie $l-l'$ in terms of set subtraction and we denote it as the \emph{drop-out AS set}), or within the ASes absent in $l$ but present in $l'$ (\ie $l'-l$ in terms of set subtraction and we denote it as the \emph{pop-up AS set}). For each anomalous prefix event, we compare its all suspicious route changes to find the intersection of their drop-out AS sets and the intersection of their pop-up AS sets, respectively. If any set is not empty, we denote the ASes in two intersection sets as the responsible ASes for the anomalous prefix event. Thus, the responsible ASes for an anomalous prefix event is defined as follows:

\begin{definition}[Responsible AS]
    Given an anomalous prefix event $A$, let $L = \{(l, l') \mid (t,r,p,p',l,l') \in A\}$. If $I_{\text{drop-out}} = \bigcap_{(l,l') \in L}(l-l') \neq \emptyset$ or $I_{\text{pop-up}} = \bigcap_{(l,l') \in L}(l'-l) \neq \emptyset$, the responsible ASes for $A$ are the ASes belonging to $I_{\text{drop-out}} \cup I_{\text{pop-up}}$. Otherwise, there are no responsible ASes for $A$.
\label{def:responsible-as}
\end{definition}

Then, the anomaly detector correlates the anomalous prefix events based on their responsible ASes. Given two different anomalous prefix events, if their time ranges are overlapped and they have common responsible ASes, we consider they are correlated. The formal definition of a correlated anomalous prefix event is:

\begin{definition}[Correlated Anomalous Prefix Event]
    Given two different anomalous prefix events $A$ and $A'$, let $T = \{t \mid (t,r,p,p',l,l') \in A\}, T' = \{t \mid (t,r,p,p',l,l') \in A'\}$, denote the set of responsible ASes for $A$ as $R_A$ and that for $A'$ as $R_{A'}$. $A$ and $A'$ are correlated if and only if $R_{A} \cap R_{A'} \neq \emptyset, \exists t_1, t_2 \in T, t_1 < t_2, t'_1, t'_2 \in T', t'_1 < t'_2, [t_1, t_2]\cap[t'_1, t'_2] \neq \emptyset$.
\label{def:correlated-anomalous-prefix-event}
\end{definition}

Finally, we divide all anomalous prefix events into different sets, where the event in one set only correlates the other events in the same set. The anomaly detector treats each set as an individual routing anomaly and outputs a corresponding alarm that includes both the affected prefixes and the responsible ASes.

\section{Alarms Validation via Pattern Matching}
\label{sec:appendix:classify-detected-routing-anomaly-events}

\update{We discuss the rationale of P1-P4 in \S\ref{sec:detection:results}, which are confirmed by domain experts. Each pattern represents one typical kind of Internet routing anomalies. Specifically, P1 tracks the change of RPKI validity states, which indicates the legitimacy of AS-Prefix bindings~\cite{mohapatra2013bgp}, during a route change to detect the unauthorized prefix announcements caused by misconfiguration or hijacking attacks. P2 tracks the violation of the valley-free criterion during a route change. This criterion describes an inherent attribute of normal routing paths in terms of business relationships~\cite{gao2001inferring}. According to RFC 7908~\cite{sriram2016problem}, the violation of this criterion indicates a route leak. P3 tracks the appearance of suspicious path segments during a route change, including private ASNs and adjacent ASes that do not have business relationship records in the CAIDA dataset. In particular, the private ASNs~\cite{mitchell2013autonomous} only appear in the global routing paths when adversaries tamper with the routing path or misconfigurations happen. Besides, considering CAIDA maintains a relatively comprehensive knowledge base of the business relationships among global ASes, the adjacent ASes that have no relationship records in the CAIDA dataset typically represent nonexistent routing connections and are very likely to be forged by sophisticated attacks, \eg the Type-N hijacking~\cite{sermpezis2018artemis}. Lastly, P4 tracks the conflicts between route announcements and Route Origin Authorization (ROA) objects where the conflicting ASes belong to the same organization. Since the ASes in the same organization usually share the same interests, these conflicts indicate the ROA object misconfigurations.}
\update{We explain how to use these patterns to analyze the alarms.}
For each reported routing anomaly event, we calculate the fraction of suspicious route changes in the event that match each pattern (\ie P1-P4) as the explanatory power (EP) of each pattern. Then, if the EP of any pattern is higher than 50\% \update{(\ie an absolute majority)}, we consider the event an anomaly with high confidence (H.C.), otherwise with low confidence (L.C.). Specifically, we use the RIPEstat web API\footnote{https://stat.ripe.net/ui2013/} to query RPKI validation states and the CAIDA AS-to-organization mapping dataset~\cite{caida_as_organization} to identify the organizations of ASes. We show the number of H.C. and L.C. events detected by our system in each dataset and the average EP of the H.C. events in Table ~\ref{tab:classification-result}.
 \update{We report the L.C. events as false alarms in \S\ref{sec:detection:results}.}


\begin{table}[t]
    \caption{\textbf{Anomalous patterns of the routing anomaly events reported by our detection system.}}
    \label{tab:classification-result}
    \small
    \centering
    \begin{tabu}{X[4,l]|X[1,c]|X[1,r]|X[1,r]|X[1,c]|X[1,r]|X[1,c]|X[2,r]}
        \hline
        \multirow{2}{*}{\textbf{Name}}&\multicolumn{5}{c|}{\textbf{H.C.}}&\multirow{2}{*}{\textbf{L.C.}}&\multirow{2}{*}{\textbf{Avg.EP}}\\\cline{2-6}
        &\textbf{P1}&\textbf{P2}&\textbf{P3}&\textbf{P4}&\textbf{Any}&& \\
        \hline
        $SP_{backcon\_5}$&0&25&30&2&32&2&0.6478 \\
        $SP_{backcon\_4}$&2&20&19&1&21&0&0.8163 \\
        $SP_{backcon\_2}$&1&29&31&3&36&1&0.7628 \\
        $SP_{bitcanal\_1}$&1&10&16&0&16&0&0.7473 \\
        $SP_{petersburg}$&2&18&24&3&24&0&0.8912 \\
        $SP_{defcon}$&3&3&3&1&6&1&0.9624 \\
        $SO_{iran}$&5&21&27&2&29&2&0.8190 \\
        $SO_{bitcanal\_3}$&8&36&38&2&39&1&0.8966 \\
        $SO_{backcon\_3}$&3&24&25&2&30&5&0.7520 \\
        $SO_{backcon\_1}$&0&15&12&1&15&3&0.8691 \\
        $SO_{bitcanal\_2}$&4&19&24&2&24&0&0.8798 \\
        $SO_{h3s}$&1&13&12&1&14&0&0.8444 \\
        $SO_{pakistan}$&2&8&5&0&9&1&0.7058 \\
        $PO_{brazil}$&4&34&50&2&50&1&0.5878 \\
        $PO_{sprint}$&1&26&28&3&29&0&0.9262 \\
        $RL_{jtl}$&1&30&34&0&41&5&0.7072 \\
        $RL_{stelkom}$&3&32&37&0&40&3&0.7644 \\
        $RL_{itregion}$&5&33&38&1&40&4&0.8187 \\
        \hline
    \end{tabu}
\end{table}

Given that some route change events occurred long ago, before their corresponding prefixes had registered ROA objects, we use the most recent RPKI validation results to identify the RPKI-related patterns P1 and P4. Further, to reduce the negative impact of the inconsistency between historical and current ROA objects, we double-check each event matching P1 and P4 based on the RIPEstat database and only regard those with unchanged authorized AS-prefix bindings as H.C.


\section{Comparison with ML-based Methods}
\label{sec:appendix:comparison}

We implement two baselines, \ie AV and LS, according to their papers~\cite{shapira2022ap2vec,dong2021isp}. Here, we describe their experimental setup and compare them to our system in detail. 

\noindent
\textbf{AV Setup.}
In accordance with the original setup in the paper~\cite{shapira2022ap2vec}, for each dataset, we train an AV detection model based on the most recent RIB snapshots collected from all vantage points before the confirmed anomaly occurs, and perform detection on the next RIB snapshots in the next two hours. \update{To reduce the training time, the model is trained offline on all two-hour RIB snapshots in parallel. Note that this strategy is not realizable for an AV model deployed for real-world online detection task, because the model has to be trained on the RIB snapshot collected 2 hours earlier and perform detection on the latest RIB snapshot.}
We set the hyperparameters 
the same to the original paper. The outputs of the AV detection model are anomalous route changes. For fair comparisons, we further apply our anomaly detector (see \S\ref{sec:detection:procedure}) to aggregate these route changes into alarms, \ie identifying anomalous prefixes and locating responsible ASes.

\noindent
\textbf{LS Setup.}
The LS method aggregates all BGP announcements in each two minutes into a time interval and extracts 86 features. Note that LS is a supervised detection method. \update{Thus, for each dataset, we use the other 17 datasets as the training data to train an LS model, and then classify the time intervals (\ie normal or anomalous) in this dataset.} The labeling method of the training data and the parameters of the LS detection model are the same to the original paper. All route changes in the detected anomalous time intervals are considered anomalous. Finally, we also utilize our anomaly detector to aggregate these anomalous route changes into alarms.

\noindent
\textbf{The Comparison.}
\update{We compare our system with AV and LS in Table~\ref{tab:comparison-baseline}. Specifically, since AV performs detection in the RIB snapshots of every two hours, it cannot detect transient routing anomalies that occur within the period between two RIB snapshots. Besides, AV requires frequent model updates because tens of millions of RIB entries are updated bihourly, which incurs significant retraining overheads. To the best of our knowledge, there is no technique that can be directly applied to accelerate the training of AV. By contrast, our system directly checks every route change collected from each real-time BGP announcement such that it can detect both transient and long-lived routing anomalies. Further, our system only requires the AS relationship data for training, which is much smaller than the RIB snapshots and more stable, ensuring a small retraining overhead.}



LS characterizes announcements as time series and applies a supervised neural network to detect if each data point in the time series is associated with routing anomalies. 
However, each data point in the time series, 
\eg the data in a two-minute time interval, 
may include a large number of legitimate announcements, leading to much more false alarms than our system.
\update{Besides, since LS utilizes a supervised model, it requires a large amount of BGP anomaly data for training, \eg more than 2,000 BGP anomaly events~\cite{dong2021isp}, which is difficult to achieve in practice. By contrast, our system performs detection in an unsupervised manner and gets rid of the reliance on anomalous training data.}


\begin{table*}[ht]
    \caption{\textbf{Comparison with ML-based methods.} \emph{Vol.}, \emph{Req.} and \emph{Ann.} are short for Volume, Requirement and BGP Announcement, respectively.}
    \label{tab:comparison-baseline}
    \footnotesize
    \centering
    \begin{tabu}{X[3,c] X[0.5,c] X[10,c] X[10,c] X[10,c] X[10,c] X[0.5,c] X[10,c] X[10,c]}
    \toprule
    \multirow{2}{*}{\textbf{Method}}&&\multicolumn{4}{c}{\textbf{Training}}&&\multicolumn{2}{c}{\textbf{Detection}}\\
    \cmidrule{3-6}
    \cmidrule{8-9}
    &&\textbf{Training Data}&\textbf{Training Data Vol.}&\textbf{Retraining Req.}&\textbf{ML Model Type}&&\textbf{Detection Data}&\textbf{Anomaly Type}\\
    \midrule
    Ours && AS relationships & \(\sim\)500K & Monthly & Unsupervised && Ann & Short \& Long-lived \\
    AV~\cite{shapira2022ap2vec} && RIB entries & \(\sim\)10M & Bihourly & Unsupervised && RIB entries & Long-lived \\
    LS~\cite{dong2021isp} && Ann. time series & \(\sim\)100M & Unknown & Supervised && Ann. time series & Short \& Long-lived \\
    \bottomrule
    \end{tabu}
\end{table*}


\section{An Empirical Study on Adaptive Attacks}
\label{sec:appendix:adaptive-attacks}

\update{
We empirically study how many ASes an adaptive attacker needs to control and how many AS relationships need to be established to make the routing roles of a malicious path\footnote{A malicious path refers to a routing path that contains the known malicious AS in the BGP anomaly event.} similar to those of a legitimate path. Specifically, for each legitimate routing path from the 18 real-world datasets (described in \S\ref{subsec:path_diff_score_results}), we start from the vantage point of the legitimate path and search for the malicious path that has the most overlapping segments with the legitimate path.
Then, we choose the ASes that appear on the malicious path but not on the legitimate path. These ASes have to be controlled by the attacker to mimic the legitimate path ASes' routing roles. Next, we compare the neighbors of the malicious and the legitimate path and calculate the minimum number of relationships that need to be established to make 10\% neighbors the same. In this study, we assume that having 10\% common neighbors is sufficient to make the routing roles of a malicious path similar to those of the legitimate path. This reduces the difficulty for the attacker, such that the result is a lower-bound estimate.


We show the result in Fig.~\ref{fig:adaptive-attack}, where the top and bottom figures display the average number of ASes that must be controlled and the average number of AS relationships that need to be established, respectively. 
The error bar shows the 95\% confidence interval. We can see that, in most cases, an adaptive attacker has to control at least two ASes and establish hundreds of new AS relationships to make the routing roles of a malicious path and a legitimate path similar. Note that, the cost for an adaptive attacker to bypass our system in practice would be much higher than the above result. This is because our system monitors multiple vantage points simultaneously, \ie the attacker has to make the routing roles of multiple malicious routing paths similar to those of normal paths. Thus, it is very difficult for adaptive attackers to bypass our system's detection.}

\begin{figure}[ht]
  \centering
    \includegraphics[width=0.9\linewidth]{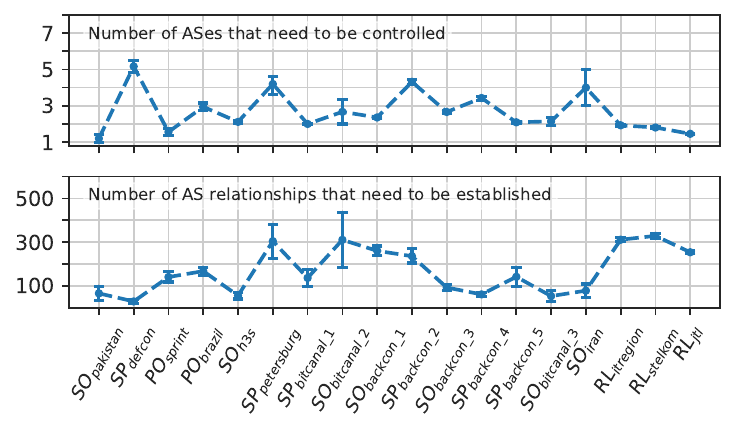}
  \caption{\textbf{\update{The empirical estimate of the cost for an adaptive attacker to make the routing roles of a malicious path and a legitimate path similar.}}}
  \label{fig:adaptive-attack}
\end{figure}